\pdfoutput=1

\newif\ifnotes
\newif\iflncs

\lncsfalse

\notestrue

\iflncs
\documentclass[orivec,runningheads,envcountsect,envcountsame]{llncs}
\usepackage{enumerate,algorithm,algorithmic}
\usepackage[normalem]{ulem}
\usepackage{paralist}
\else
\documentclass{article}
\usepackage{fullpage}
\usepackage{amsmath, amsthm, amssymb}

\fi

\usepackage[dvipsnames]{xcolor}
\definecolor{DarkBlue}{RGB}{0,0,150}
\usepackage[colorlinks,linkcolor=DarkBlue,citecolor=DarkBlue]{hyperref}

\usepackage{xspace}
\usepackage{url}
\usepackage{amsmath,amssymb,enumerate,algorithm,algorithmic,latexsym}
\usepackage{breakcites}

\usepackage{braket}
\usepackage{tikz}
\usetikzlibrary{fpu}
\usepackage{footmisc}
\usepackage{float}
\newfloat{algorithm}{H}{lop}
\usepackage{color}
\usepackage{colortbl}
\usepackage{graphicx}
\usepackage{subfigure}
\usepackage{amsfonts}
\usepackage{enumitem}
\usepackage{soul}
\usepackage[normalem]{ulem}
\usepackage{bookmark}
\numberwithin{algorithm}{section}

\renewcommand{\paragraph}[1]{\vspace{1.5mm}\noindent \textbf{#1}}

\iflncs
\fi

\newcommand{\bb}{\mathbb}

\newcommand{\abs}[1]{\left|#1\right|}

\newcommand{\pST}{\; \middle\vert \;}

\newcommand{\zo}{\{0,1\}}

\newcommand{\secp}{\lambda}

\def\binset{\zo}

\newcommand{\poly}{\mathrm{poly}}

\newcommand{\Sim}{\mathcal{S}}

\newcommand{\tr}{\mathsf{tr}}

\newcommand{\bbE}{\mathbb{E}}

\newcommand{\bbN}{\mathbb{N}}

\newcommand{\cA}{\mathcal{A}}
\newcommand{\cB}{\mathcal{B}}

\newcommand{\cH}{\mathcal{H}}

\newcommand{\cM}{\mathcal{M}}

\newcommand{\cP}{\mathcal{P}}
\newcommand{\cQ}{\mathcal{Q}}
\newcommand{\cR}{\mathcal{R}}
\newcommand{\cS}{\mathcal{S}}

\newcommand{\cY}{\mathcal{Y}}
\newcommand{\cZ}{\mathcal{Z}}

\newcommand{\textabbrevstyle}[1]{\mbox{#1}}
\newcommand{\textabbrevstylebol}[1]{\mbox{\textbf{#1}}}
\newcommand{\newtextabbrev}[1]{\expandafter\newcommand\csname #1\endcsname{\textabbrevstyle{#1}\xspace}}
\newcommand{\newtextabbrevbol}[1]{\expandafter\newcommand\csname #1\endcsname{\textabbrevstylebol{#1}\xspace}}
\newcommand{\renewtextabbrevbol}[1]{\expandafter\renewcommand\csname
	#1\endcsname{\textabbrevstylebol{#1}\xspace}}

\newtextabbrevbol{EXP}
\newtextabbrevbol{Dtime}
\newtextabbrev{PRG}
\newtextabbrev{PRF}
\newtextabbrev{OWF}
\newtextabbrev{PPT}
\newtextabbrev{QPT}
\newtextabbrev{EOWF}
\newtextabbrev{GEOWF}
\newtextabbrev{ZAP}
\newtextabbrevbol{Ntime}
\newtextabbrevbol{coAM}
\newtextabbrevbol{NP}
\newtextabbrevbol{TFNP}
\newtextabbrev{NIWI}
\newtextabbrevbol{AM}
\newtextabbrev{SNARG}
\newtextabbrev{SNARK}
\newtextabbrev{coRP}
\newtextabbrev{BPP}
\newtextabbrev{NIUA}
\newtextabbrev{UA}
\newtextabbrev{ZK}
\newtextabbrev{AI}
\newtextabbrev{ECRH}
\newtextabbrev{PIR}
\newtextabbrev{WI}
\newtextabbrev{SWI}
\newtextabbrev{WIPOK}
\newtextabbrev{SWIAOK}
\newtextabbrev{POK}
\newtextabbrev{FE}
\newtextabbrev{IO}
\newtextabbrev{XIO}
\newtextabbrev{SXIO}
\newtextabbrev{SKFE}
\newtextabbrev{PKFE}
\newtextabbrev{PPRF}
\newtextabbrev{LWE}
\newtextabbrev{PKE}

\iflncs
\spnewtheorem{construction}[theorem]{Construction}{\bfseries}{\itshape}
\spnewtheorem{assumption}[theorem]{Assumption}{\bfseries}{\itshape}
\spnewtheorem{Informal Theorem}[theorem]{Informal Theorem}{\bfseries}{\itshape}
\spnewtheorem{fact}[theorem]{Fact}{\bfseries}{\itshape}
\spnewtheorem{clm}[theorem]{Claim}{\bfseries}{\itshape} %
\else

\newtheorem{theorem}{Theorem}[section]
\newtheorem{definition}[theorem]{Definition}
\newtheorem{lemma}[theorem]{Lemma}
\newtheorem{corollary}[theorem]{Corollary}
\newtheorem{claim}[theorem]{Claim}
\newtheorem{clm}[theorem]{Claim}
\newtheorem{proposition}[theorem]{Proposition}
\theoremstyle{remark}
\newtheorem{remark}{Remark}[section]

\fi

\newenvironment{boxfig}[2]{\begin{figure}\fbox{\begin{minipage}{\linewidth}
				\vspace{0.2em}
				\makebox[0.025\linewidth]{}
				\begin{minipage}{0.95\linewidth}
					{{
							#2 }}
				\end{minipage}
				\vspace{0.2em}
	\end{minipage}}}{\end{figure}}

\ifnotes
\newcommand{\nir}[1]{$\ll$\textsf{\color{red} Nir: { #1}}$\gg$}
\newcommand{\zvika}[1]{$\ll$\textsf{\color{blue} Zvika: { #1}}$\gg$}
\newcommand{\yael}[1]{$\ll$\textsf{\color{purple} Yael: { #1}}$\gg$}
\else
\newcommand{\nir}[1]{}
\newcommand{\zvika}[1]{}
\newcommand{\yael}[1]{}
\fi

\newcommand{\Ex}{\mathop{\bbE}}

\newcommand{\ketbra}[1]{\ket{{#1}}\bra{{#1}}}

\newcommand{\td}{\mathop{\mathrm{TD}}}
\newcommand{\pk}{\mathsf{pk}}
\newcommand{\sk}{\mathsf{sk}}
\newcommand{\tk}{\mathsf{tk}}
\newcommand{\Gen}{\mathsf{Gen}}
\newcommand{\TGen}{\mathsf{TokenGen}}
\newcommand{\Sig}{\mathsf{Sig}}
\newcommand{\TSig}{\mathsf{TokenSig}}
\newcommand{\ver}{\mathsf{Ver}}

\usepackage[utf8]{inputenc}

\title{Constructive Post-Quantum Reductions}
\author{}
\author{Nir Bitansky\thanks{Tel Aviv University, Israel, \texttt{nbitansky@gmail.com}. Member of the checkpoint institute of information security. Supported by ISF grants 18/484 and 19/2137, by Len Blavatnik and the Blavatnik Family Foundation, by the Blavatnik Interdisciplinary Cyber Research Center at Tel Aviv University, and by the European Union Horizon 2020 Research and Innovation Program via ERC Project REACT (Grant 756482).} \and Zvika Brakerski\thanks{Weizmann Institute of Science, Israel,  \texttt{zvika.brakerski@weizmann.ac.il}. Supported by the Israel Science Foundation (Grant No.\ 3426/21), and by the European Union Horizon 2020 Research and Innovation Program via ERC Project REACT (Grant 756482) and via Project PROMETHEUS (Grant 780701).} \and Yael Tauman Kalai\thanks{Microsoft Research and Massachusetts Institute of Technology, MA, USA, \texttt{yael@microsoft.com}.  This material is based upon work supported by DARPA under Agreement No.\ HR001120200.  Any opinions, findings and conclusions or recommendations expressed in this material are those of the author(s) and do not necessarily reflect the views of the United States Government or DARPA.
 }
}
\date{}

\begin{document}

\maketitle

\begin{abstract}
Is it possible to convert classical reductions into post-quantum ones? It is customary to argue that while this is problematic in the interactive setting, non-interactive reductions do carry over. However, when considering quantum auxiliary input, this conversion results in a \emph{non-constructive} post-quantum reduction that requires duplicating the quantum auxiliary input, which is in general inefficient or even impossible. This violates the win-win premise of provable cryptography: an attack against a cryptographic primitive should lead to an algorithmic advantage.

We initiate the study of constructive quantum reductions and present positive and negative results for converting large classes of classical reductions to the post-quantum setting in a constructive manner. We show that any non-interactive non-adaptive reduction from assumptions with a polynomial solution space (such as decision assumptions) can be made post-quantum constructive. In contrast, assumptions with super-polynomial solution space (such as general search assumptions) cannot be generally converted.

Along the way, we make several additional contributions:
\begin{enumerate}
	\item We put forth a framework for reductions (or general interaction) with \emph{stateful} solvers for a computational problem, that may change their internal state between consecutive calls. We show that such solvers can still be utilized. This framework and our results are meaningful even in the classical setting.
	
	\item A consequence of our negative result is that quantum auxiliary input that is useful against a problem with a super-polynomial solution space cannot be generically ``restored'' post-measurement. This shows that the novel rewinding technique of Chiesa et al.\ (FOCS 2021) is tight in the sense that it cannot be extended beyond a polynomial measurement space.
\end{enumerate}

\end{abstract}
\thispagestyle{empty}
\newpage
\tableofcontents
\thispagestyle{empty}
\newpage

\newcommand{\qaux}{\ket{s}}

\section{Introduction}

The notion of provable security in cryptography has had a great impact on the field and has become a de-facto gold standard in evaluating the security of cryptographic primitives. A provably secure cryptographic primitive is stated in the form of a computational problem $P$, whose hardness is related by means of \emph{reduction} to that of another problem $Q$ which is either by itself considered intractable or in turn can be further reduced down the line. The reduction is an algorithm that solves the problem $Q$ provided that it is given access to an algorithm that solves the problem $P$.

This gives rise to the ``win-win principle'' which stands as one of the main motivations for using provably secure cryptography. The logic is the following. Either an algorithmic solution for $P$ cannot be found, i.e.\ the cryptographic primitive $P$ is secure for all intents and purposes, 
or one can find an algorithmic solution for $P$ which would imply an algorithmic solution for $Q$, thus contributing to the state of the art in algorithms design. 
Indeed, cryptographic reductions are the main working tool for the theoretical cryptographer. Numerous reductions between cryptographic primitives are known and hundreds of such reductions are published in the cryptographic literature every year.

The emergence of the quantum era in computing poses a new challenge to provable security and the win-win principle. Many existing reductions in the ``pre-quantum'' world implicitly or explicitly relied on the $P$-algorithm being classical. These reductions are thus a-priori invalid when considering quantum algorithms. A central line of investigation in the domain of post-quantum security is thus dedicated to the following question.
\begin{center}
	\emph{To what extent can pre-quantum reductions be ported to the post-quantum setting?}
\end{center}

Such conversion may not always be possible. This is particularly a concern when considering \emph{interactive} problems, i.e.\ ones where the solution to $P$ involves multiple messages being exchanged with the solver algorithm. Indeed, one of the most prominent techniques for proving security in the interactive setting, namely the notion of \emph{rewinding}, does not directly translate to the quantum setting and moreover one can explicitly show cases where pre-quantum reductions exist but post-quantum ones do not. In fact, this property was actually used to construct \emph{proofs of computational quantumness} \cite{BCMVV18} in which a party proves that it is quantum by succeeding in a task for which there is a classical impossibility result (under computational assumptions). In a nutshell, the reason is that a quantum algorithm may keep a quantum state between rounds of interaction, and this quantum state is measured and thus potentially destroyed in order to produce the next message of interaction. It is therefore not possible to naively ``rewind'' the interaction back to a previous step as is customary in many classical proofs.

The focus of this work, therefore, is on \emph{non-interactive cryptographic assumptions}. These are problems~$P$ whose syntax contains a (randomized) instance generator which generates some instance $x$, and a verifier that checks whether solutions $y$ are valid (with respect to $x$ or more generally the randomness that was used to generate $x$). The role of the solver algorithm in this case is simply to take $x$ as input and produce a $y$ that ``verifies well'' (we avoid getting into the exact formalism at this point).

Contrary to the interactive case, it is customary to postulate (often without proof) that classical reductions to non-interactive cryptographic assumptions carry over straightforwardly to the post-quantum setting since there is no rewinding. There is a simple challenge-response interface that on the face of it ``does not care'' whether the underlying $P$-solver is implemented classically or quantumly. This viewpoint, however, is overly simplistic, since the $P$ solver may use \emph{quantum auxiliary input}: a quantum state $\qaux$ that is used as a resource for solving $P$. The state $\qaux$ can be the result of some natural process upon which we have no control, or a result of some exhaustive preprocessing, or generated in the course of execution of some protocol. At any rate, the means to produce $\qaux$ are often not at our disposal, we just get a copy of the state. 

In this case, similarly to the interactive setting, the quantum state is measured whenever the $P$-solver is called, and therefore, it potentially precludes us from calling the $P$ solver more than once. This issue is often addressed in the literature by noticing that providing many copies of $\qaux$ would allow to call the $P$ solver multiple times -- namely there exists a quantum state $\qaux^{\otimes t}$ that allows to solve $Q$ given access to the $P$ solver. Therefore, the existence of a classical reduction still implies that if $Q$ is intractable \emph{even given arbitrary auxiliary input}, then the same holds for $P$.

We argue that the aforementioned common ``solution'' for post-quantum reductions in the presence of quantum auxiliary input is unsatisfactory. First and foremost, this solution strictly violates the win-win principle. While the argument above indeed implies that (some form of) intractability for $P$ follows from (some form of) intractability for $Q$, it \emph{does not} allow to convert an auxiliary-input algorithm for $P$ into an auxiliary-input algorithm for $Q$ in a constructive manner, since the transformation $\qaux \to \qaux^{\otimes t}$ is not an efficient one. An additional related concern is the \emph{durability} of such reductions. Namely, that if we wish to execute the reduction more than once (i.e.\ solve multiple instances of $Q$) then we need to duplicate the state $\qaux$ an a-priori unbounded number of times.

Given this state of affairs, the question we are facing is the following.
\begin{center}
	\emph{To what extent can pre-quantum reductions to non-interactive assumptions be ported to the post-quantum setting constructively and durably?}
\end{center}
Naturally, we do not wish to redo decades of cryptographic work in re-proving each result individually. Instead, we would like to identify the broadest class of pre-quantum reductions that can be generically converted into the post-quantum regime, and at the same time characterize the limitations where such generic conversion is not possible. This is the focus of this work, and indeed we show a generic transformation for a very broad class of reductions. Along the way we develop an adversarial model for \emph{stateful} adversaries that may be of interest in its own right, even in the \emph{classical} setting.

\subsection{Our Main Results}

We prove a general positive result for converting classical reductions into post-quantum ones. In particular we consider {\em non-adaptive reductions.} In such reductions, the set of queries to the oracle is determined before any query is made. It turns out that an important parameter in our positive as well as our negative result is the size of the \emph{solution space} of the computational problem $P$ (``the cryptographic primitive''). Our positive results apply to cases where the solution space is polynomial. One notable example the case where $P$ is a ``decision assumption'', namely the $P$ solver is a distinguisher that returns a single bit as output. Another notable example is the case where $P$ is an NP search problem, with unique solutions (e.g., injective one-way functions or unique signatures). An informal result statement follows. 

\begin{theorem}[Positive result, informal]
	There exists an efficient transformation for converting any classical non-adaptive black-box reduction from assumption $Q$ to assumption $P$, where $P$ is a non-interactive assumption with a polynomial solution space, into a  \emph{constructive and durable} 
	post-quantum reduction from $Q$ to $P$.
\end{theorem}

We prove a complementary negative result, for the case where $P$ has a large solution space. The negative result relies on the existence of classical indistinguishability obfuscation which is secure against quantum adversaries.

\begin{theorem}[Negative result, informal]
	Assume the existence of post-quantum secure indistinguishability obfuscation. Then there exist non-interactive assumptions $P$, $Q$, where $P$ has a super-polynomial solution space and the following hold. There exists a classical non-adaptive black-box reduction from assumption $Q$ to assumption $P$, but there is no such constructive post-quantum reduction.
\end{theorem}

As explained above, in order to address the question of constructiveness, we need to develop a new adversarial model and a host of tools to address this question. An account of these intermediate contributions appears in the technical overview below.

\subsection{Our Techniques and Additional Contributions}

Known approaches fall short of achieving constructiveness and durability since they regard quantum auxiliary input similarly to its classical counterpart, despite the inherent difference of the inability to duplicate or reuse quantum information. We assert that the process of making multiple calls to an algorithm with quantum side information is inherently \emph{stateful}. Namely, the internal state of the ``oracle'' changes and evolves over time. In this work we put forth a framework for stateful solvers, namely algorithms that change their internal state and thus their behavior over time.

In the post-quantum setting, reductions start from \emph{one-shot} solvers. That is, ones that have an initial state that allows them to provide an answer for a single instance of $P$ successfully, but afterwards all bets are off. It seems natural (and, as we show, turns out to be useful) to consider stateful solvers that propagate their $P$-solving property throughout an execution, we call this property \emph{persistence}. Persistent solvers evolve their state in an arbitrary way subject to being able, at any point in their evolution, to successfully answer a $P$-query (with some noticeable advantage).

\paragraph{A Framework for Stateful Solvers.} Section~\ref{sec:modeling} is dedicated to formally defining the notion of a (potentially stateful) solver and quantifying its success probability in solving a problem $P$. We accordingly provide definitions for a post-quantum reduction in this setting, and more specifically the notion of a post-quantum black-box reduction. The standard notion of a classical black-box reduction is recovered as a special case of our definition, when specializing to so-called \emph{stateless} $P$-solvers. 

Using our new formalism, the task at hand is to convert a reduction that expects to be interacting with a \emph{stateless} solver, into one that is successful even when given a \emph{one-shot stateful} solver.

\paragraph{One-Shot Solvers Imply Persistent Solvers.} One-shot solvers may seem quite useless, since on the face of it they may only successfully respond to a single query. However, our first technical result, in Section~\ref{sec:restoration}, is that they can in fact be converted generically (but in a non-black-box manner) into persistent solvers. Namely, ones that can answer an \emph{a-priori unbounded} number of queries and maintain roughly the same success probability. The persistent solver has a state of length that is polynomially related to that of the one-shot solver. The running time of the persistent solver increases with each query it is being asked. That is, the time complexity of answering the $t$-th query scales with $\poly(t)$ for a fixed polynomial. This still ensures that for any polynomial-length sequence of queries, the total time to answer all queries is bounded by a fixed polynomial. The persistent value of the resulting solver (i.e.\ the value that is maintained for an a-priori unbounded number of times) is itself a random variable that is determined during the conversion process. The expectation of the persistent value is equal to the one-shot value of the solver we start from. (We note that it is inherently impossible to achieve a non-probabilistic behavior, i.e.\ to ensure a persistent value that is always above some threshold.\footnote{To see this, consider the case where the one-shot auxiliary input $\qaux$ is a superposition giving weight $\sqrt{1-\varepsilon}$ to a value $\ket{\bot}$  that always makes the $P$-solver fail, and giving weight $\sqrt{\varepsilon}$ to a state that makes the $P$-solver perfectly successful. Then, by trace-distance considerations, any processing of $\qaux$ must be $\varepsilon$-statistically-indistinguishable from a case where $\qaux = \ket{\bot}$. Therefore, with probability at least $1-\varepsilon$ the persistent value will be trivial. Nevertheless, using a Markov argument, if we start from a one-shot solver with a non-negligible advantage, we recover, with a non-negligible probability, a many-shot solver with a non-negligible persistent advantage.})

Our transformation is an extension of the techniques in the recent work of Chiesa, Ma, Spooner and Zhandry~\cite{CMSZ21}, that can be interpreted as showing such a transformation for ``public-coin'' cryptographic assumptions (ones where the instances are uniformly distributed and the verification requires only the instance and the solution, and not the randomness that was used to generate the instance). It is only in this step that we have the restriction that the solution space of the problem needs to be polynomial, due to limitations of the \cite{CMSZ21} technique. Our negative result (further discussed below) proves that these limitations are inherent.

The conversion from one-shot to persistent is the only transformation that uses the solver in a non-black-box manner. In the rest of our (positive) results we take a persistent $P$-solver and a bound on the length of its auxiliary quantum state and only make black-box use of this solver, i.e., provide instances as input and receive solutions as output.  We do not further intervene with the evolution of the state between consecutive calls to the solver.

Once we transformed our solver to being persistent, we are guaranteed that we can make multiple $P$ queries, and each one will be answered by a ``successful'' solver. It may seem that our mission is complete. However, this is far from being the case. While all queries are answered by a successful solver, these solvers may be arbitrarily correlated. For example, thinking about a simple linearity test where a reduction queries $x_1, x_2, x_3 = x_1\oplus x_2$ and checks whether a linear relation holds. It may be the case that for each query $x_i$ we get a response $y_i$ from an approximately-linear function, and yet the solver ``remembers'' that $x_1, x_2$ were previously made as queries, and deliberately fails on $x_1 \oplus x_2$ in the next query. Another example, that will be quite useful to illustrate our transformation is that of the Goldreich-Levin (GL) hardcore bit \cite{GL89}, where queries take the form $(f(x), r_i)$, always with the same $f(x)$, and with additional correlations between the $r_i$ values across different queries. In particular, it may be the case that once a query with some value $f(x)$ has been made, the solver refuses to meaningfully answer any additional queries with the same $f(x)$.\footnote{We note that while the \emph{classical} GL reduction, falls under our umbrella of non-adaptive reductions, in this specific case, it is in fact known how to devise a single-query \emph{quantum} reduction~\cite{AC02}. This, however, does not resolve the question of durability, and more importantly does not provide a general framework for all non-adaptive reductions.}

We note that attributing adversarial behavior to the solver is done for purposes of analysis. Our transformation from one-shot to persistent appears quite ``innocent'' and we do not know whether it can actually generate such pathological behavior that will prevent reductions from running. However, we cannot rule it out and therefore we consider a worst-case adversarial model.

When described in this way, it seems that only very specialized reductions can be carried over to the post quantum setting. For example, ones that employ a strong form of random self reduction when making solver queries. One such case is the search-to-decision reduction for the learning with errors problem \cite{Regev05}. However, as the GL hardcore bit example demonstrates, this doesn't even extend to all search to decision reductions. We must therefore find a new way to utilize stateful solvers. Indeed, the handle that we use is that while the solver may change its behavior adversarially, its adversarial behavior is constrained by the length of the auxiliary state $\qaux$ that it uses. We will indeed leverage the fact that this state is polynomailly bounded to limit the adversarial powers of the solver and handle more general reductions.

Before moving on to describe our techniques in this context, we notice that while this adversarial model (of black-box access to a persistent solver) emerged as a by-product of our work on quantum reductions, it is nevertheless a valid model in its own right in both the quantum and classical setting. We may consider interacting with an adversary/solver that is \emph{only} guaranteed to be noticeably successful at every point in time but, unlike the standard notion of an ``oracle'', may change its behavior over time. In our case, we allow the behavior to change arbitrarily, so long that the amount of information carried over between executions is bounded (in our case, by the length of the state, which is polynomially bounded).

\paragraph{Memoryless Persistent Solvers.} Our next step, in Section~\ref{sec:memoryless}, is to show that a persistent solver, even with adversarial behavior, can be effectively converted into a more predictable form of solver that we call \emph{memoryless} (note that this is different from our final goal which is to achieve a \emph{stateless} solver). A memoryless solver keeps track of the sequence number of the question it is asked (e.g.\ it knows that it is now answering query number $4$) but it is not allowed to remember any information about the actual content of the previous queries that were made.

We show that a combination of a non-adaptive reduction and a persistent solver induce a memoryless (persistent) solver (more accurately a distribution over memoryless solvers). These memoryless solvers are accessible using a simulator that, given access to the reduction and the original solver, efficiently simulates the interaction of the reduction with the induced memoryless solver, up to inverse-polynomial statistical distance. 
Note that we require that the reduction is non-adaptive. Namely, its queries to the solver can be arbitrarily correlated (as in the GL case), but the identity of the queries must not depend on the answers to previous queries.

The transformation relies on the fact that the solver has a bounded amount of memory, say $\ell$ qubit of state that is propagated through the execution. Our strategy is to dazzle the solver with an abundance of i.i.d dummy queries, that are sampled from the marginal distribution of the ``real'' queries (for example, in the GL case, each dummy query will have the form $(f(x_i), r_i)$ where $x_i, r_i$ are both random). In between the dummy queries, in random locations, we plant our real queries, in random order. We prove that the solver, having only $\ell$ qubits of state, must answer our real queries as if they were dummy queries. This requires us to develop a proper formalism and to prove a new lemma (Plug-In Lemma) using tools from quantum information theory. See Section~\ref{apx:pluginproof} for the full details.

\paragraph{Stateless Solvers at Last.} Finally, we show in Section~\ref{sec:memlesstostateless} that memoryless solvers imply stateless solvers. This is again shown by means of simulation via a similar formalism to the previous result. Recall that a stateless solver must answer all queries according to the same distribution. This transformation again relies on the non-adaptive nature of the reduction, namely on the ability to generate all solver-queries ahead of time. To do this, we notice that we can think of a memoryless solver simply as a sequence of stateless solvers that can be queried one at a time. Therefore, we can consider the induced stateless solver that at every query picks a random solver from this collection and executes it on the query. This indeed will result in a stateless solver. The solving probability of the induced stateless solver is simply the average success probability of solvers in the collection, which is concentrated due to persistence. Moreover, this behavior can be simulated by randomly permuting the \emph{queries}, while still calling the solvers according to their order in the sequence.\footnote{Remember that we have access to the memoryless solver which only allows to make queries in order.} 

This way, asking the queries in a permuted order to the memoryless solver will (almost) mimic the action of sampling a solver from the collection independently for each query. The only reason why this mimic is not perfect is that permuted queries are sampling ``without repetition'', i.e.\ none of the solvers in the sequence defined by the memoryless solver will be queried twice, whereas in the ideal strategy we described above, it is possible that the same solver from the sequence will be sampled more than once. We deal with this by making the number of solvers in the sequence so big, that the probability of hitting the same solver twice becomes very small (inversely polynomial for a polynomial of our choice). We simply add to our queries of interest a large number of dummy ``0 queries'', and perform a random permutation on this extended set of queries. %

\paragraph{Putting Things Together.} In Section~\ref{sec:classicalna} we put all of the components together and prove our main positive result, that any classical non-adaptive reduction which relies on a non-interactive polynomial-solution-space assumption can be made post quantum. This requires putting together the components in a careful manner.

The fact that the first step in our transformation was to produce a persistent $P$, allows us to continue using it even after having solved a $Q$ instance. In particular, this means that we can solve additional instances of $Q$, or use it to solve additional instances of $P$ or any other problem $Q'$ for which a non-adaptive reduction to $P$ exists. In particular, this property implies that our reduction is \emph{durable}.

\paragraph{A Negative Result for Search Assumptions.} We show in Section~\ref{sec:search}, that a generic conversion from classical to constructive quantum reductions is not always possible, even for the case of non-adaptive reductions to non-interactive assumptions. In particular, if $P$ is an assumption with a large solution space (intuitively, a search assumption) this may not be possible.

We show our negative result by relying on a recently introduced primitive known as tokenized signatures \cite{Ben-DavidS16}. These are signature schemes with the standard classical syntax, but for which it is possible to produce a quantum \emph{signature token}. The signature token allows to generate a single classical signature for a message of the signer's choice, but only one such signature can be created. Tokenized signatures have been constructed relative to a classical oracle \cite{Ben-DavidS16} or based on cryptographic assumptions \cite{ColadangeloLLZ21}.

We can define an ``assumption'' which is essentially the task of signing a random message using a tokenized signature scheme.\footnote{The assumption is instantiated by a verification key which we can think of as non-uniformity of the assumption, see discussion in Section~\ref{sec:search}.} In the classical world, there is a trivial reduction between the task of signing one random message and the task of signing two random messages. However, if we consider a quantum solver that holds the token as auxiliary input, then by definition it should not be possible to use it to obtain two signatures for two different messages. Our negative result holds for any conversion process that is constructive, and in particular does not obtain any implicit non-uniform advice about the assumption.

\subsection{Other Related Work}

The question of which reductions can be translated from the classical to the post-quantum setting also received significant attention in the context of the random-oracle model (ROM), starting from the work of Boneh et~al.~\cite{BonehDFLSZ11}. The question asked in these works is whether it is possible to convert reductions in the classical ROM into ones the quantum ROM (QROM, where the adversary is allowed to make quantum queries to the oracle).
There are several results proving that specific schemes that are secure in the ROM are also secure in the QROM~\cite{Zhandry12,TarghiU16,JiangZCWM18,Katsumata0Y18,Zhandry19,DonFMS19,LiuZ19,DonFM20,KramerS20}.  Recently, a more general ``lifting theorem" was given in \cite{YamakawaZ21}, showing how to convert a proof in the ROM to one in the QROM for any ``search-type game" where a challenger makes only a {\em constant} number of queries to the random oracle.  This work also presented a negative result, showing that there are schemes that are secure in the ROM yet are insecure in the QROM.
While the general motivation in these works is similar to ours, the question they ask is quite different from ours. In the ROM/QROM, the solver is allowed to make queries to the oracle (which is simulated by the reduction), which is more similar to the setting where interactive-assumptions are used.

Our memoryless transformation (Section~\ref{sec:memoryless}) relies heavily on the state of the solver being bounded in length. The idea that bounded quantum memory can be used to restrict an otherwise all powerful adversary is at the core of the bounded quantum storage model. It can be shown (see, e.g., \cite{DamgardFSS08}) that it is possible to achieve cryptographic abilities against strong adversaries while relying only on a limit on the amount of quantum storage they can use. This setting is quite different from ours, though, since the quantum bounded storage model allows an unbounded amount of classical memory, which in our setting would make it impossible to achieve any result. Indeed, the bounded storage model requires quantum communication (whereas our reduction-solver communication is completely classical), and thus the set of tools and techniques that are used in both settings are completely different.

\section{Preliminaries and Tools}\label{sec:prel}

\newcommand{\qv}[1]{\boldsymbol{{#1}}}
\newcommand{\qsts}[1]{\mathbf{S}(#1)}

We say that a given function $f(x_1,\dots,x_k)$ is  $\poly(x_1,\dots,x_k)$, if there exist constants $c,C$ such that $(x_1\cdot x_2\cdot\ldots\cdot x_k)^c\leq f\leq (x_1\cdot x_2\cdot\ldots\cdot x_k)^C$.

We denote by $\td$ the trace distance between two matrices. 

\paragraph{Algorithms.} By default, when referring to an {\em algorithm} we mean a classical probabilistic (resp. quantum) algorithm. Algorithms may be uniform or non-uniform, meaning that they have {\em classical} advice related to the input size (we specify when uniformity matters). An {\em efficient} algorithm is also polynomial time. 

\paragraph{Quantum Notation.}  We use standard quantum information in Dirac notation. We denote quantum variables in boldface $\qv{x}$ and classical variables in lowercase $x$. The density matrix of $\qv{x}$ is denoted $\rho_{\qv{x}}$. Classical variables may also have (diagonal) density matrices. Quantum variables $\qv{x}, \qv{y}$ have a joint density matrix $\rho_{\qv{x}, \qv{y}}$ if they can be jointly produced by an experiment. As usual, $\qv{x}, \qv{y}$ are independent if $\rho_{\qv{x}, \qv{y}} = \rho_{\qv{x}}\otimes\rho_{\qv{y}}$. We never assume that quantum variables are independent unless we explicitly say so.
Quantum registers are denoted in capital letters. We also sometimes use capital letters to denote distributions, where it is clear from the context.
For a finite Hilbert space $\cH$ we denote by $\qsts{\cH}$ the set of density matrices over quantum states in $\cH$.

A quantum procedure is a general quantum algorithm that can apply unitaries, append ancilla registers in $0$ state, perform measurements in the computational basis and trace out registers. The complexity of $F$ is the number of \emph{local} operations it performs (say, operations on up to $3$ qubits are considered local). If $F$ is a quantum procedure then we denote by $F(\qv{x})$ the application of $F$ on $\qv{x}$.
   Any unitary induces a quantum procedure that implements this unitary, which does not perform measurements or trace out registers, we call this procedure ``a unitary quantum circuit''.

\paragraph{Purification of Quantum Procedures and States.} 
A quantum procedure may introduce new ancilla qubits, perform intermediate measurement throughout its computation and discard registers or parts thereof. However, any quantum procedure can be \emph{purified} into unitary form without much loss in complexity \cite{QCQI}. This is formally stated below.

\begin{proposition}\label{prop:purification}
	Let $C$ be a general quantum procedure of complexity $s$. Then it is possible to efficiently generate a \emph{unitary} quantum circuit  $\widehat{C}$ of size $O(s)$, such that for any quantum state $(\qv{x}, \qv{a})$,  setting $(\qv{y}, \qv{z}) = \widehat{C}(\qv{x}, \qv{0})$, it holds that $(\qv{y},\qv{a})$ has identical density matrix to $(C(\qv{x}), \qv{a})$.
\end{proposition}

Likewise, any quantum state can be viewed as a reduced density matrix of the output of a unitary (which may be inefficient to implement)  . 
\begin{proposition}\label{prop:purifystate}
	Let $\qv{x}$ be a variable with density matrix $\rho_{\qv{x}}$. Then there exists a unitary $U$ over registers $XY$ such that applying $U(\qv{0},\qv{0})$, the reduced density matrix of the value in the $X$ register has density matrix $\rho_{\qv{x}}$.
\end{proposition}

\subsection{The Plug-In Lemma}
\label{sec:plugin}

The following lemma is another manifestation of information incompressibility in the quantum setting. Specifically, we are interested in an experiment in which an all powerful compressing procedure attempts to compress $t$ samples which are arbitrarily distributed into $\ell$ quantum bits. We show that this is infeasible even in the weak sense in which a decoder receives the compressed value, and a $(j-1)$-prefix of the sequence, and is required to identify the $j$-th element. We show that as $t$ increases, the probability of succeeding in the experiment drops. A formal statement follows.

\begin{lemma}[Plug-In Lemma]\label{lem:plugin}
	Let $\vec{Y} = (Y_1, \ldots, Y_t)$ be a joint distribution over $t$ classical random variables. Let $\vec{y}$ be distributed according to $\vec{Y}$. Let $\qv{s}$ be an $\ell$-qubit random variable that has arbitrary dependence on $\vec{y}$.	
	We let $\vec{y}_i$ denote the prefix $\vec{y}_i = (y_1, \ldots, y_i)$ for $1 \le i \le t$, and $\vec{y}_0$ is the empty vector (and likewise for $\vec{Y}$).
	Let $J$ be the uniform distribution over $[t]$ and let $j \gets J$. Define $y' \gets Y_J | (\vec{Y}_{j-1} = \vec{y}_{j-1})$.
	Then it holds that
	\begin{align}
		\td( (j, \vec{y}_{j-1}, y_j, \qv{s}), (j, \vec{y}_{j-1}, {y}', \qv{s}) ) \le \sqrt{\ell/(2t)}~.
	\end{align}
\end{lemma}
Note that the above two distributions are \emph{not} identical even though $(j, \vec{y}_{j-1}, y_j)$ and $(j, \vec{y}_{j-1}, y')$ are identically distributed. The reason is that in both cases, $\qv{s}$ is always generated as a function of $\vec{y}$, i.e.\ using $y_j$ and not $y'_j$.

 The lemma is proven in Section~\ref{apx:pluginproof}.

\newcommand{\atg}{\mathsf{a}}
\newcommand{\val}{\mathsf{val}}
\newcommand{\stt}{\mathsf{state}}
\newcommand{\hstt}{\widehat{\mathsf{state}}}

\section{Assumptions, Stateful Solvers, and Reductions}
\label{sec:modeling}

In this section, we formally define the concepts of non-interactive cryptographic assumptions, stateful solvers, and their value and advantage in breaking an assumption.

\subsection{Non-Interactive Assumptions}
We define the notion of a non-interactive (falsifiable) cryptographic assumption as in \cite{Naor03,HaitnerH09}. While we frame the notion as ``cryptographic'', it can be viewed more generally as a notion for average-case problems where the solution can be verified.

\begin{definition}[Non-Interactive Assumption]\label{def:assump}
	A non-interactive assumption is associated with polynomials $d(\secp),n(\secp),m(\secp)$ and a tuple $P = (G,V,c)$ with the following syntax. The \emph{generator} $G$ takes as input $1^\secp$ and $r \in \binset^d$, it returns $x \in \binset^n$. The \emph{verifier} $V$ takes as input $1^\secp$ and $(r, y) \in \binset^d \times \binset^m$ and returns a single bit output. (Both $G$ and $V$ are deterministic.) $c(\secp)$ is the assumption's \emph{threshold.}

	We say that $P$ is \emph{falsifiable} if $G,V$ are uniform polynomial-time algorithms (in their input size). 
\end{definition}

We also define a property called {\em verifiably-polynomial image} that roughly speaking requires that any instance has at most polynomial many solutions and that this can be verified in some weak sense. The property in particular captures problems where the solution space $\zo^m$ is of polynomial size such as decision problems (where $m=1$), and problems in $\NP$ where there are a few solutions per instance (such as injective one-way functions).

\def \imver {K}
\begin{definition}[Verifiably-Polynomial Image]\label{def:ver_poly_im}
A non-interactive assumption $P$ has a \emph{verifiably-polynomial image} if there exists an efficient verifier $\imver$ and a polynomial $k=\poly(\secp)$, such that for any instance $x \in \zo^n$, the set $Y_x :=\{y: \imver(1^\secp,x,y)=1\}$ of $\imver$-valid solutions is of size at most $k$ and for any valid instance $x=G(1^\secp,r)$ and solution $y$ such that $V(1^\secp,r,y)=1$, it holds that $y\in Y_x$.
\end{definition}

\noindent
The traditional notion of the advantage in solving an assumption $P$ is measured in terms of the distance between the solving probability (which we term the value) and the threshold $c$.

\begin{definition}[Value and Advantage of Classical Functions]\label{def:val_adv_state_classic}
Let $P=(G,V,C)$ be a non-interactive assumption and let $f = \set{f_\secp:\zo^n\rightarrow\zo^m}_\secp$ be a family of (possibly randomized) functions. For every $\secp \in \bbN$, we define the corresponding value and advantage:
\begin{align*}
  \val_P[f](\secp) := \Pr\left[V(1^\secp, r, y)=1 \pST \begin{array}{c} r \gets \zo^d\\
                  x = G(1^\secp,r)\\
                    y \gets f_\secp(x) \end{array}\right] \iflncs,\else\hspace{1cm}\fi \atg_P[f](\secp) := \abs{\val_P[f](\secp) - c(\secp)}\enspace,
\end{align*}
where the probability is also above the randomness of $f_\secp$ in case it is randomized.
 \end{definition}

\subsection{Stateful Solvers}
The premise of our work is that in the quantum setting, one ought to think about \emph{stateful} solvers, which generalizes the standard treatment of a solver as a one-shot algorithm. We now define this formally.

\begin{definition}[Stateful Solvers: Syntax]\label{def:solvers}
	Let $P$ be a non-interactive assumption. 
		
		\item Let $\ell=\ell(\secp)$ be a function. A classical (resp.\ quantum) \emph{$\ell$-stateful solver} $\cB = (B, \stt_0=\{  \stt_{\secp,0}\}_\secp)$ is defined as follows.  
		
		\begin{itemize}
			\item $B$ is a classical (resp.\ quantum) algorithm that takes as input $1^\secp$, $1^t$, $x \in \binset^n$ and $\stt$ which is an $\ell$-bit (resp.\ qubit) string, and outputs a value $y \in \binset^m$ and $\stt'$ which is an $\ell$-bit (resp.\ qubit) next-state. We let $B(\cdots)_\mathrm{y}$ denote the $y$ output and $B(\cdots)_\mathrm{st}$ denote the $\stt'$ output.
						
			\item $\stt_0 =\{ \stt_{\secp,0}\}_\secp$ is a sequence of classical (resp.\ quantum) states consisting of $\ell=\ell(\secp)$ bits (resp.\ qubits).
		\end{itemize}
		We say that $\cB$ is efficient if $B$ runs in time $\poly(\secp,t,n)$; i.e., in polynomial time in the lengths of its inputs. 
	
\end{definition}

\begin{remark}[Non-uniformity]
The algorithm $B$ may have a non-uniform classical advice. It does not have any additional quantum advice.
\end{remark}

\begin{remark}[Dependence on Runtime]
Our definition allows the running time of efficient stateful solvers to depend polynomially on the ``iteration'' $t$. In particular, for any polynomial number of solving attempts $t=\poly(\secp)$, the overall running of the solver is polynomial. One could also consider a more stringent definition that requires that each call runs in fixed polynomial time independently of the iteration number~$t$. Jumping forward, we will show how a solver can preserve its solving ability through time, but at the cost of running for longer in each step. Doing this according to the more stringent time-independent definition remains an open question.
\end{remark}

It will be useful to define some properties of solvers with respect to an extension of the sovler's execution transcript. The extension corresponds to the would-be transcript of a purified version of the solver, running on a purified version of the initial state. This will allow us to get a precise well-defined handle on the evolution of quantum states throughout the lifetime of the solver.  The extended transcript will only be used for purposes of definition and analysis and will never be required algorithmically.  %

\def\hy{\hat{y}}
\def\hY{\hat{Y}}
\def\ts{\mathsf{ts}}
\def\hts{\widehat{\ts}}

\def\ovp{\widehat{P}}
\def\ovb{\widehat{B}}
\def\ova{\widehat{A}}
\def\ovcb{\widehat{\cB}}
\def\ovig{\widehat{\instg}}
\def\ovv{\widehat{V}}
\renewcommand{\Sim}{\mathsf{SimMemless}}
\newcommand{\simsl}{\mathsf{SimStateless}}
\newcommand{\simall}{\mathsf{Sim}}

\def \hcB {\widehat{\cB}}
\def \hB {\widehat{B}}
\begin{definition}[Stateful Solvers: Purifying
Values]\label{def:solversext}
Consider a solver $\cB = (B, \stt_0=\{  \stt_{\secp,0}\}_\secp)$.
Let $B_{\secp, t, x}$ denote the quantum procedure that takes $\qv{s}$ as input and produces $B(1^\secp, 1^t, x, \qv{s})$ over registers $SY$. By Proposition~\ref{prop:purification}, we can consider its purification $\ovb_{\secp, t, x}$ which acts on registers $SY\hY$ and takes as input $(\qv{s},\qv{0},\qv{0})$. Then define $\ovb(1^\secp, 1^t, x, \qv{s})$ as the algorithm that computes $(\qv{s}', \qv{y}, \qv{\hy}) = \ovb_{\secp, t, x}(\qv{s},\qv{0},\qv{0})$, measures $(\qv{y}, \qv{\hy})$ in the computational basis to obtain $(y,\hy)$, and then outputs $\qv{s}'$ as the $\stt$ output, $y$ as the solution output, and $\hy$ as the \emph{purifying output}.

  In addition, by Proposition~\ref{prop:purifystate}, there exists a (possibly inefficient) unitary  $\ovb_{0,\secp}$ that operates on two registers $S\hY$ such that when applying $(\qv{s}_0, \qv{\hy}_0)\gets \ovb_{0,\secp}(\qv{0}, \qv{0})$, the reduced density matrix of $\qv{s}_0$ is identical to that of $\stt_0$. 
Then define $\ovb_0(1^\secp)$ as the quantum procedure that computes $(\qv{s}_0, \qv{\hy}_0)\gets \ovb_{0,\secp}(\qv{0}, \qv{0})$, measures $\qv{\hy_0}$ in the computational basis, and then outputs $\qv{s}_0$ as $\stt_0$ and $\hy_0$ as the \emph{purifying initial value}.

We refer to the collection $\hcB = \set{\hB_{i,\secp,x}}$   as {\em a purification of $\cB$} (it is not unique).
\end{definition}
\begin{remark}
We note that the purifying values can be arbitrarily long. These values will only be used for analysis purposes and are never produced in an actual execution, and hence we do not require any bound whatsoever on the length of the purifying values or the complexity of producing them.
\end{remark}

\medskip
\noindent

We now define the concept of a \emph{solver interaction,} which captures the process of repeatedly invoking a stateful solver by a given algorithm.

\begin{definition}[Solver Interaction]\label{def:solver_interact}
	Let $P=(G,V,c) $ be a non-interactive assumption. For any stateful solver $\cB = (B, \stt_0)$ and corresponding purification $\hcB$, and any algorithm $A$ with input $z \in \zo^*$, we consider the process $A^\cB(1^\secp,z)$ of the algorithm interacting with the solver. We define this process in two different yet equivalent manners: one which is efficient given the ability to execute $B$, and one which may be inefficient but implies an identical output distribution. The latter will include a production of all purifying values (Definition~\ref{def:solversext}) which will be useful for definitions and analysis.
	\begin{itemize}
	
	\item We let $\stt_0$ be as defined in $\cB$. \\ Equivalently: We let $(\stt_0, \hy) \gets \ovb_0(1^\secp)$.
	\item
	$A$ is invoked on input $(1^\secp,z)$ and at every step $i\geq 1$:
	\begin{enumerate}
	    \item $A$ submits a query $x_i \in \zo^n$.
	    \item  %
	    $(y_i,\stt_i) \gets B(1^\secp, 1^i, x_i, \stt_{i-1})$ is invoked. \\ Equivalently: $(\hy_i, y_i,\stt_i) \gets \ovb(1^\secp, 1^i, x_i, \stt_{i-1})$ is invoked.  
	    \item  
	    $A$ obtains $y_i$, and proceeds to the next step.
	\end{enumerate}
	\item
	At the end of the interaction $A$ may produce an output $w$.
	\end{itemize}
	 We sometimes refer to $A$ as a {\em solver-aided} algorithm and use the shorthand $A_z^\cB$ for the solver interaction and $A_z^{\hcB}$ for the purified solver interaction. 
	 We refer to the random variables $\stt_0,\stt_1,\stt_2,\dots$ as the state random variables of the interaction. We refer to the list of pairs of generated instances and solutions $(x_i,y_i)$ as the transcript of the interaction and denote it by $\ts$. We also define the \emph{extended transcript} $\hts$ of the execution as consisting of the value $\hy_0$ followed by a list to triples $(x_i,y_i,\hy_i)$. Given an extended transcript $\hts$, we can produce the standard transcript $\ts$ by removing all purifying values. We call this action \emph{redaction} and say that $\ts$ is the redacted transcript induced by $\hts$. Generating an extended transcript according to the purified solver interaction $A_z^{\hcB}$ and redacting it produces an identical distribution to the generation of the redacted transcript by direct interaction $A_z^\cB$. The length of a transcript/extended-transcript is the number of pairs/triples it contains (this means that an extended transcript of length $0$ is not empty since it still contains $\hy_0$ . The $i$-prefix of a transcript/extended-transcript is denoted $\ts_i$/$\hts_i$ and contains the first $i$ pairs/triples (and also $\hy_0$ in the extended case).
\end{definition}

We show that the purifying values indeed purify the entire solver interaction, in the sense that they determine all states $\stt_i$ as pure states for any solver interaction.

\begin{proposition}\label{prop:exttsstate}
	Let $\cB = (B,\stt_0)$ be a solver with purification $\hcB$ and consider the extended transcript $\hts$ of the solver interaction $A_z^{\hcB}$ and let $t$ be its length. Then for all $i \le t$, the state $\stt_i$ is pure conditioned on $\hts_i$. Specifically, it has density matrix $\ketbra{s_{\hts_{i}}}$ that is completely determined by $\hts_i$ (and therefore by the classical string $\hts$) and does not depend on any other parameter of the execution.
\end{proposition}
\begin{proof}
	We consider the purifying description of the solver interaction $A_z^{\hcB}$ and prove by induction.
	For $t=0$, we recall that the pair $(\stt_0, \hy)$ is generated by applying $\ovb_{0,\secp}$ on the zero state, followed by measuring the $\hY$ register. The pre-measurement state over registers $S\hY$ is therefore pure, and can always be written as
	\begin{align}
		\sum_{\hy} \alpha_{\hy}  \ket{s_{\hy}}_S \otimes \ket{\hy}_{\hY}~,
	\end{align}
		where $\alpha_{\hy}$ are non-negative real values with $\sum_{\hy} \alpha_{\hy}^2 = 1$, and $\ket{s_{\hy}}$  are fully specified unit vectors.
	Therefore, post-selecting on having measured the value $\hy_0$ in register $\hY$, we have that the state in register $S$ is exactly $\stt_0=\ketbra{s_{\hy_0}}$, which completes the base step of the proof.
	
	Now assume that the above holds for all $i < t$. Consider a transcript $\hts$ of length $t$ s.t.\ $\hts = \hts_{t-1} \| (x,y,\hy)$ for some $\hts_{t-1}, x ,y, \hy$.
	
	Let us consider the state of the system right before the $t$-th query to the solver. At this point, $\hts_{t-1}$ was already determined, and thus by induction we know that $\stt_{t-1}=\ketbra{s_{\hts_{t-1}}}$ is a pure state. At this point $x$ has also been determined. 
	
	By definition, $\stt_t$ is produced by executing a unitary $\hB_{\secp, t, x}$ (that acts on registers $SY\hY$) on $(\stt_{t-1}, \qv{0}, \qv{0})$, which is pure by the induction hypothesis, and measuring the $Y\hY$ registers. The analysis here is similar to the base case. The pre-measurement state is pure (since it is induced by applying a unitary on a pure state) and thus can always be written as
	\begin{align}
		\sum_{y,\hy} \alpha_{y,\hy}  \ket{s_{y,\hy}}_S \otimes \ket{y, \hy}_{Y\hY}~,
	\end{align}
	and as above $\alpha_y$ are non-negative real values with $\sum_y \alpha_y^2 = 1$, and $\ket{s_{y,\hy}}$ are fully specified unit vectors. Post selecting on $y,\hy$ leaves us with register $S$ containing $\stt_t=\ketbra{s_{y,\hy}}$, which completes the proof.
\end{proof}

	We are now ready to define the concepts of value and advantage of stateful solvers. Traditionally, when thinking about stateless solvers, we consider their {\em one shot value}, namely the probability that they solve the problem on a random instance. Since they are stateless this probability does not change over time. In the case of stateful solvers, this probability may change over time. Our definition of the {\em many shot values} aims to capture exactly this. For any solver interaction $A_z^\cB$, the value at time $t$, captures the probability that the solver $\cB$ successfully solves a random instance  at this time, after a given $t$-round   interaction with $A_z$. This value is, in fact, a random variable that depends on the history of the interaction. To make this precise, we consider any purification $\hcB$, and define these values as a function of the extended transcript.

\begin{definition}[Stateful Solvers: Value and Advantage]\label{def:val_adv_state}
Let $P$ be a non-interactive assumption, $\cB = (B,\stt_0)$ be a corresponding stateful solver, $\hcB$ a corresponding purification, and $A$ a solver-aided algorithm with input $z$. For every $\secp, i \in \bbN$, let $\stt_i$ be the $i$-th pure state random variable of the solver interaction $A_z^{\hcB}$ (determined by $\hts_i$). The corresponding value random variables are:
\begin{align*}
  \val_P[i,A_z^{\hcB}](\secp) &:= \Pr\left[V(1^\secp, r, y)=1 \pST \begin{array}{c} r \gets \zo^d\\
                  x = G(1^\secp,r)\\
                  (\hy_{i+1},y,\stt_{i+1}) \gets \ovb(1^\secp,1^i,x,\stt_i)   \end{array}\right]\enspace,
                    \end{align*}
where the probability is over the choice of $r$ and the measurement of $\hy_{i+1},y$.

\smallskip\noindent
The one-shot value of $\cB$ is
\begin{align*}
\val_P[0,\cB](\secp) :=\Pr\left[V(1^\secp, r, y)=1 \pST \begin{array}{c} r \gets \zo^d\\
                  x = G(1^\secp,r)\\
                 (y,\stt_{1}) \gets B(1^\secp,x,\stt_0) \end{array}\right]\enspace,
                    \end{align*}    
where the probability is over the choice of $r$, measurements of $B$, and (the possibly mixed) $\stt_0$. Note that this is in fact a number, independent of any $A$ or the choice of purification $\hcB$.

The corresponding advantage random variables are:
                    \begin{align*}
 \atg_P[i, A_z^{\hcB}](\secp) &:= \abs{\val_P[i, A_z^{\hcB}](\secp) - c(\secp)}\hspace{1cm}\atg_P[0,\cB](\secp) := \abs{\val_P[0,\cB](\secp) - c(\secp)}\enspace.
\end{align*}
For a distribution $\bb{B}$ on solvers $\set{\cB_\alpha}_\alpha$, we define the one-shot value of the distribution as:
$$
\val_P[0,\bb{B}](\secp) = \bb{E}_{\alpha\gets \bb{B}}[\val_P[0,\cB_\alpha](\secp)]\enspace.
$$
The corresponding advantage is $\atg_P[0,\bb{B}](\secp) = |c(\secp)-\val_P[0,\bb{B}](\secp)|$.

 \end{definition}
As the solver's state evolves over time, its advantage in solving an assumption may reduce or disappear altogether. This is in particular relevant to the quantum setting, where when a solver is invoked its internal state is disturbed. Aiming to capture solvers that remain useful over time, we next define the notion of solvers with {\em persistent value,} namely, solvers whose value in solving a given assumption is preserved through time. We define it more generally for distributions over solvers; single solvers are a special case.

\begin{definition}[Persistent Value]\label{def:pers_val}
Let $P$ be a non-interactive assumption. A distribution $\bb{B}$ on solvers $\set{\cB_\alpha}_\alpha$ is $\eta$-persistent if there exist purifications $\set{\hcB_\alpha}_\alpha$   such that for any algorithm $A$ with input $z$, with probability $1-\eta$ over the choice of solver $\alpha\gets \bb{B}$ and over an extended transcript $\hts$ in the solver interaction process $A_z^{\hcB_\alpha}$, there exists a value $p$  such that:
    \begin{align}\label{eq:persistent}
    \max_{ i}\left|\val_P[i,A_z^{\hcB_\alpha}]- p\right| \leq \eta \enspace.
    \end{align}
    We call $p$ {\em a persistent value}. Given a random variable $p^*(\alpha) \subseteq [0,1]$, we say that a solver is $(p^*,\eta)$-persistent if the condition holds for $p^*(\alpha)$. 
\end{definition}

We next define the notion of a persistent advantage. This aims to capture the case that solvers maintain a lower bound on their advantage through time.

\begin{definition}[Persistent Advantage]
 Let $P$ be a non-interactive assumption with threshold $c$. A distribution $\bb{B}$ on solvers $\set{\cB_\alpha}_\alpha$ has $\varepsilon$-persistent advantage if there exist purifications $\set{\hcB_\alpha}_\alpha$ such that for any algorithm $A$ with input $z$:
    \begin{align}\label{eq:persistentadv}
    \bb{E}\left[\min_{ i} \val_P[i,A_z^{\hcB_\alpha}]\right] \geq c+\varepsilon \enspace,
    \end{align}
    where the expectation is over the choice of solver $\alpha\gets \bb{B}$ and over an extended transcript $\hts$ in the solver interaction process $A_z^{\hcB_\alpha}$.
 \end{definition}
 
 In the above, We require that the advantage has a consistent sign (for simplicity, positive). Intuitively, the reason we focus on persistence of the positive advantage $v_t-c$ at time $t$, rather than the absolute advantage $|v_t-c|$, is that if the sign of $v_t-c$ arbitrarily changes after each solver invocation, then the solver may not be as useful. (As a simple example, take a deterministic distinguisher and turn it into a stateful distinguisher that flips the output of the original distinguisher at random with each invocation, deeming it useless.) We note that $\eta$ persistent solvers in particular preserve the sign of their advantage (up to $\eta$).

 \paragraph{Memoryless and Stateless Solvers.} 	A special case of the above definitions is that of {\em memoryless and stateless solvers.}
 
 \begin{definition}\label{def:memless_stateless}
 A solver $\cB=(B,\stt_0)$ is {\em memoryless} if the size of its state is $\ell=0$. The solver is \emph{stateless} if in addition (to being memoryless), the algorithm $B$ does not depend on $1^t$ (in functionality or runtime). 
 \end{definition}
 
 \begin{remark}[Persistent Value for Stateless and Memoryless Solvers]\label{rmk:persistentmemless}
 Note that in the case of stateless solvers, successive invocations of the solver will always result in the same output distribution. Here the one-shot (and many-shot) advantage coincide with the standard notion of advantage for functions (Definition \ref{def:val_adv_state_classic}) and values are persistent (Definition \ref{def:pers_val}). Accordingly, stateless solvers exactly capture the traditional notion of classical solvers, given by a randomized function. 
 
 Moreover, even for memoryless solvers, when considering the definition of persistent solvers the value $\val_P[i,A_z^{\hcB}]$ does not depend on $A_z$ at all (only on $i$), and therefore it is a fixed number rather than a random variable. It follows that for $(p,\eta)$-persistent memoryless solvers, Eq.~\eqref{eq:persistent} holds with probability $1$.
 \end{remark}

\subsection{Reductions} 

We now define the notion of a reduction. A reduction is a way to prove a claim of the form ``if there exists a successful solver for assumption $P$ then there exists a successful solver for an assumption $Q$''. We consider \emph{constructive reductions} in the sense that they are an explicit uniform algorithm that takes as input a successful solver for $P$ and efficiently solves the problem $Q$. 

The default notion of a reduction in the literature is {\em one shot}. In pparticular, a given quantum $P$-solver is only assumed to have a meaningful one-shot advantage in solving $P$, and there is no a priori guarantee on its advantage in any many shot solving process, in particular there may not be any value persistence.  Likewise, the produced solver for the assumption $Q$ is only required to have a meaningful one-shot advantage. Below we define both the default notion of one-shot reductions as well as the stronger notion of {\em durable reductions} requiring that the resulting $Q$-solver also has persistent advantage, meaning that with noticeable probability, the reduction can go on solving for an arbitrary polynomial number of times.

\begin{definition}[Reduction]\label{def:reduction}
A reduction from classically (resp.\ quantumly) solving a non-interactive assumption $Q$ to classically (resp.\ quantumly) solving a non-interactive assumption $P$ is an efficient classical (resp. quantum)  uniform algorithm $\cR$ with the following guarantee. For any solver $\cB_P = (B_P, \stt_0)$ for $P$ with one-shot advantage $\varepsilon$ and running time $T$, let $\stt'_0 = (\stt_{0},B_P, 1^{1/\varepsilon},1^{T})$. Then $\cB_Q = (\cR,\stt'_0)$ is a solver for $Q$ with one-shot advantage $ \varepsilon' = \poly(\varepsilon,T^{-1},\secp^{-1})$ and running-time $\poly(T,\varepsilon^{-1},\secp)$. We say that the reduction is durable if $\cB_Q$ has $\poly(\varepsilon,T^{-1},\secp^{-1})$-persistent advantage.

We refer to a reduction from solving $Q$ to classically (resp. quantumly) solving $P$ as a {\em classical-solver (resp. quantum-solver) reduction.}
\end{definition}

\begin{remark}[Many Shot Reductions]
    There could be several conceivable extensions of the above definition that also account for the {\em many-shot advantage}. One such natural extension is requiring that the reduction works only given a solver with a persistent value (as in Definition \ref{def:pers_val}). Jumping ahead, in section \ref{sec:restoration}, we show that under certain conditions, persistent solving can in fact be reduced to one-shot solving, even in the quantum setting. %
\end{remark}

\begin{remark}[The Loss]
    We allow for a (fixed) polynomial loss in the advantage and running time. One could naturally extend it to more general relations.
\end{remark}

\paragraph{Classical Black-Box Reductions.} In this work, we prove that several general classes of {\em classical reductions} that a priori are only guaranteed to work for classical solvers, can be enhanced {\em efficiently} to also work for quantum solvers. Our focus is on black-box reductions; that is, reductions that are oblivious of the representation and inner workings of the solver that they use (in contrast to the above Definition \ref{def:reduction}, where the reduction obtains the full description of the solver $\cB_P$).

We next formally define such black box reductions, using the terminology we have already developed. Specifically, we capture the notion of a classical solver for a given problem $P$ as a stateless (classical) solver.

\begin{definition}[Classical Black-Box Reduction]\label{def:classicalbbred}
A classical black-box reduction, from solving a non-interactive assumption $Q$ to solving a non-interactive assumption $P$, is an efficient classical solver-aided uniform algorithm $\cR$ with the following syntax and guarantee. $\cR$ takes as input a security parameter $1^\secp$, parameter $1^{1/\varepsilon}$, and instance $x\in \zo^{n_Q}$ of $Q$. It interacts with a solver $\cB$ for $P$ (per Definition \ref{def:solver_interact}) and produces an output $y\in \zo^{m_Q}$. We require that for any distribution $\bb{B}$ over {\bf stateless classical} solvers $\set{\cB_\alpha}_\alpha$ such that $\bb{B}$ has advantage at least $\varepsilon$ in solving $P$,  the corresponding solver distribution $\bb{R}$ over solvers $\set{\cR^{\cB_\alpha}(1^\secp,1^{1/\varepsilon},\cdot)}_\alpha$ has advantage at least $\poly(\varepsilon,\secp^{-1})$ in solving $Q$. The advantage of $\cR$ is positive if its value is always at least $c_Q$ (above the assumption $Q$'s threshold), regardless of any $P$-solver.  

We further say that the reduction $\cR$ is \emph{non-adaptive} if $\cR$ produces all of its oracle queries $x_1,\dots,x_k \in \zo^{n_P}$ to $\cB$ in one shot, obtains all answers $y_1,\ldots,y_k$, and then produces its output $y$.
\end{definition}

\begin{remark}
In our definition of solver interaction, a given solver $\cB$ is only ever invoked for the instance size $n_P(\secp)$. Accordingly, the above definition restricts attention to classical reductions that in order to solve problem $Q$ for instance size $n_Q(\secp)$ make queries to a $P$-solver on a specific related input size $n_p(\secp)$. While this is not without loss of generality, it does capture natural reductions. (In fact, we are not aware of important reductions that do not adhere to this.)
\end{remark}

\begin{remark}[Deterministic Solver Reductions, Positive Advantage, and Repeated Queries]\label{rem:clas_red_wlog}
We consider classical reductions that ought to work when given a stateless solver from a distribution $\bb{B}$ over solvers $\set{\cB_\alpha}$. (As a matter of fact in our model, even once a stateless solver $\cB_\alpha$ is fixed, the process of answering any given query is randomized, but this can be modeled as sampling a deterministic stateless solver from another distribution $\bb{B}$ with the same advantage.) A weaker notion of classical reductions only requires that the reduction works for deterministic solvers. In the classical setting, this is typically not an issue, as long as the reduction has the power to fix the solver's randomness and repeatedly replace it as needed. Jumping forward, when considering quantum reductions, the randomness of a given solver may arise from the quantum nature of the solving process, and the reduction may not be able to control it. Accordingly, in our transformations from classical-solver reductions to quantum-solver reductions, we will naturally need the classical reduction we start from to also be able to deal with distributions over solvers.

We note that for typical assumptions $Q$ such as search problems (with trivial threshold $c=0$) or decision problems (with solution length $m=1$ and trivial threshold $c=1/2$), a classical reduction $\cR$ from $Q$-solving to deterministic $P$-solving implies a classical reduction $\cR$' from $Q$-solving to distributional $P$-solving. Here two subtleties should be addressed. The first issue that could prevent $\cR$ from working for distributional $P$-solvers is that the sign of the advantage of $\cR^\cB$ as a $Q$-solver may depend on the randomness of $\cB$ and may cancel out in expectation. For search assumptions $Q$, where $c=0$, this cannot happen as any advantage is positive. For decision problems, this can be avoided by slightly augmenting $\cR$ to make sure that the advantage is always positive using standard black-box techniques \cite{BrakerskiG11}. This incurs only a polynomial overhead in solving queries, or even just a single query, at the cost of quadratically decreasing the advantage. The second issue concerns the running time of the reduction. Specifically a reduction that works for deterministic oracles, excepts to get their advantage $1^{1/\varepsilon}$ as input, where $\varepsilon$ is the $P$-solver's advantage. When executing such a reduction with a solver distribution, we are given $1^{1/\varepsilon}$, where $\varepsilon$ is the average advantage. Nevertheless, we can run the original reduction with input $1^{2/\varepsilon}$. Note that the probability that that the advantage of a sampled oracle is at least $\varepsilon/2$ is at least $\varepsilon/2$, and since the reduction has positive advantage, we are overall guaranteed to maintain a noticeable advantage.  

Following the above, for typical assumptions $Q$, we can in particular assume w.l.o.g positive advantage. For simplicity, we also assume throughout that classical reductions We do not repeat queries. This is w.l.o.g as given a deterministic oracle, the reduction can simply store previous answers and answer consistently by itself.%
\end{remark}

\section{Persistent Solvers in the Quantum Setting}\label{sec:restoration}

In this section, invoking state restoration techniques from \cite{CMSZ21}, we prove that any one-shot solver for an assumption $P$ with a verifiably-polynomial image (in particular, decision problems) can be converted into a persistent solver for $P$.

\begin{theorem}[Persistence Theorem]\label{thm:state_restoration}
Let $P$ be a non-interactive falsifiable assumption with a verifiably-polynomial image. For any inverse polynomial function $\eta$, there exist efficient quantum algorithms $S,R$ with the following syntax and guarantee. $S_B(\stt_0)$ takes as input a quantum algorithm $B$ and state $\stt_0$ and outputs a state $\stt^*_0$ and a value $p^*\in [0,1]$. $R_B(1^\secp,1^i,x,\stt_{i-1}^*)$ takes as input $B$, a security parameter $1^\secp$, step $1^i$, input $x\in\zo^n$, and state $\stt_{i-1}^*$ and outputs a solution $y\in\zo^m$ and state $\stt^*_{i}$. 

For any solver $\cB = (B,\stt_0)$ with one-shot value $p = \val_P[0,\cB]$, considering the random variable $(\stt_0^*,p^*) \gets S_B(\stt_0)$, it holds that:
\begin{enumerate}
  \item
  $\bbE \left[p^* \right] = p$.
    \item  $\cR^* = (R_B,\stt_0^*)$ sampled in this process is a distribution over efficient stateful solvers that is $(p^*,\eta)$-persistent.
\end{enumerate}
\end{theorem}

\begin{remark}
The efficiency of the algorithms $S,R$ is also polynomial in the running time of $B$. We avoid passing explicitly the running time bound as input to simplify notation.
\end{remark}

The proof is based on two lemmas from \cite{CMSZ21}, adapted to our notation and simplified for our needs. %

\def \valest {\mathsf{ValEst}}
\def \hvalest {\widehat{\mathsf{ValEst}}}
\def \repair {\mathsf{Repair}}
\begin{lemma}[Lemmas 4.9 and 4.10 in \cite{CMSZ21}, adapted]\label{lem:csmz}
	Let $\cH$ be a Hilbert space. There exist efficient quantum algorithms:
	\begin{enumerate}
		\item $(\qv{\rho}^*,p^*)\gets \valest_{V,A}(\qv{\rho},1^{1/\varepsilon})$ that given as input any verifier circuit $V:\zo^d \times \zo^n \rightarrow \zo$, any quantum circuit $A$, a quantum state $\qv{\rho}\in \qsts{\cH}$ and accuracy parameter $1^{1/\varepsilon}$, outputs a quantum state $\qv{\rho}^*\in \qsts{\cH}$ and value $p^* \in [0,1]$;
		\item
		$\qv{\sigma}^*\gets \repair_{V,A,\Pi}(\qv{\sigma},y,p,1^{1/\varepsilon},1^k)$ that given $V$, $A$, and a projective measurement $\Pi = (\Pi_y)_{y\in Y}$ on $\cH$ with outcomes $Y=\set{y_1,\dots,y_k}
		$, and given a quantum state $\qv{\sigma}\in \qsts{\cH}$, outcome $y \in Y$, probability $p \in [0,1]$, and parameters $1^{1/\varepsilon}$, $1^k$, outputs a quantum state $\qv{\sigma}^*\in \qsts{\cH}$;
	\end{enumerate}
	such that the following guarantees hold:
	\begin{enumerate}
		\item {\bf Value Estimation:} $$\bbE_{}\left[p^*\pST (\qv{\rho}^*,p^*) \gets \valest_{V,A}(\qv{\rho},1^{1/\varepsilon})\right] = \Pr\left[V(r,y) = 1\pST \begin{array}{l}
			r\gets \zo^d  \\
			y \gets A(\qv{\rho},r) 
		\end{array}\right]\enspace.$$
		
		\item {\bf Estimation is Almost Projective:} 
		
		$$\text{For $\varepsilon\geq\varepsilon'>0$,} \hspace{5mm}\Pr\left[|p^* - p^{**}| \geq \varepsilon \pST \begin{array}{l}
			(\qv{\rho}^*,p^*) \gets \valest_{V,A}(\qv{\rho},1^{1/\varepsilon})\\
			(\qv{\rho}^{**},p^{**}) \gets \valest_{V,A}(\qv{\rho}^*,1^{1/\varepsilon'})
		\end{array}\right] \leq \varepsilon \enspace.$$
		
		\item {\bf Repairing:} 
		$$\text{For $\varepsilon >0$,} \hspace{5mm}\Pr\left[|p^* - p^{**}| \geq \varepsilon \pST \begin{array}{l}
			(\qv{\rho}^*,p^*) \gets \valest_{V,A}(\qv{\rho},1^{1/\varepsilon})\\
			(\qv{\sigma},y ) \gets \Pi(\qv{\rho}^*)\\
			\qv{\sigma^*} \gets \repair_{V,A,\Pi}(\qv{\sigma},y,p^*,1^{1/\varepsilon},1^k)\\
			(\qv{\rho}^{**},p^{**}) \gets \valest_{V,A}(\qv{\sigma}^*,1^{1/\varepsilon})
		\end{array}\right] \leq \varepsilon \enspace.$$
	\end{enumerate}
\end{lemma}

\begin{remark}[On the Restriction to Verifiably-Polynomial Image Assumptions]
We note that the reason that the Persistence Theorem \ref{thm:state_restoration} is restricted to assumptions $P$ with a verifiably-polynomial image stems from the fact that the running time of the repairing procedure $\repair$ given by Lemma \ref{lem:csmz} scales with the number of outcomes $k$ of the corresponding projection $\Pi$.  It is  tempting to search for a better repairing procedure that does not scale with~$k$, since (as we show) it would imply a stronger version of Theorem \ref{thm:state_restoration}, without the restriction of a verifiably-polynomial image.  However, as we show in Section~\ref{sec:search} (Theorem \ref{thm:explicit_reduction}), such a stronger version of Theorem \ref{thm:state_restoration} does not exist (provided that  $R,S$ use the assumption $P$ as a black box,\footnote{Or in some more general explicit fashion that we define in Section~\ref{sec:search}.} which is indeed the case in our proof). 
\end{remark}

\medskip\noindent
We now proceed to prove the Persistence Theorem \ref{thm:state_restoration} relying on Lemma \ref{lem:csmz}.

\def \bB {\widetilde{B}}

\begin{proof}[Proof of Theorem \ref{thm:state_restoration}.]
	Throughout, fix the assumption $P$, with generator $G$, verifier $V$, and polynomial-image verifier $\imver$ with corresponding polynomial image bound $k$. Also fix
	an inverse polynomial function $\eta$ and the security parameter $\secp$.

	\paragraph{Notation and Conventions.} In what follows, we consider a one-shot quantum solver $y \gets B(1^\secp,x,\stt)$ that operates on quantum states $\stt\in \qsts{\cH}$. We assume that $\cH = \cZ\cY$, where for every $x \in\zo^n$, there is an efficiently computable unitary purification $\hB_{\secp,x}$ on $\cH$, describing the action of $B(1^\secp,x,\stt)$ on $\stt$, where $\cY$ is the output register (this assumption regarding $\cH$ is w.l.o.g by polynomially increasing the size of the state if needed ). We also consider a {\em wrapper solver} $\bB(1^\secp,r,\stt^*)$ that given $\stt^*\in \qsts{\cH}$, computes $x = G(1^\secp,r)$, applies $\hB_{\secp,x}$, and outputs a measurement of $\cY$ in the computational basis.
	
	For any $x\in \zo^n$, we denote by $\Pi_x = (\Pi_{x,y})_{y\in Y_x \cup \bot}$ the projective measurement on $\cH$ where: 
	\begin{align*}
		Y_x &:= \set{y\in \zo^m: \imver(1^\secp,x,y)=1} \enspace,\\[3mm]
		\Pi_{x,y} &:= \hB_{\secp,x}^\dag (I_{\cZ} \otimes \ketbra{y}_{\cY})\hB_{\secp,x} \hspace{5mm}\text{for $y\in Y_x$}\enspace,\\[3mm]
		\Pi_{x,\bot} &:= \sum_{y\in \zo^m\setminus Y_x} \hB_{\secp,x}^\dag( I_{\cZ} \otimes  \ketbra{y}_{\cY})\hB_{\secp,x} \enspace.
	\end{align*}

	\medskip\noindent
	\item
	We define the algorithms $S$ and $R$. Throughout, the algorithms $\valest$ and $\repair$ operate on corresponding states $\stt^*$ in $\qsts{\cH}$.  
	
	\medskip\noindent
	$S_B(\stt_0)$:
	\vspace{-1mm}
	\begin{itemize}
		\item     
		Output $(\stt_0^*,p^*_0) \gets \valest_{V,\bB}(\stt_0,1^{2/\eta})$.
	\end{itemize}
	
	\medskip\noindent
	$R_B(1^\secp,1^i,x_i,\stt_{i-1}^*)$:
	\vspace{-1mm}
	\begin{itemize}
		\item
		Let $\varepsilon_i = \eta/(i \pi)^2$. 
		\item     
		Apply $(\qv{\sigma}_{i-1},p_{i-1}) \gets \valest_{V,\bB}(\stt_{i-1}^*,1^{1/\varepsilon_i})$.
		\item     
		Apply $(\qv{\sigma}_i^*,y_i)  \gets \Pi_{x_i}(\qv{\sigma}_{i-1})$.
		\item     
		Apply $\qv{\rho}_i  \gets \repair_{V,\bB,\Pi_{x_i}}(\qv{\sigma}_i^*,y_i,p_{i-1},1^{1/\varepsilon_i},1^{k+1})$.
		\item
		Apply $(\stt_i^*,p_i^*) \gets \valest_{V,\bB}(\rho_i,1^{1/\varepsilon_i})$.
		\item
		Output $(y_i, \stt_i^*)$
	\end{itemize}
	We now prove that the above algorithms satisfy the requirements of Theorem \ref{thm:state_restoration}. 
	
	First note that if $B$ has one-shot value $p$, then by the estimation guarantee of Lemma \ref{lem:csmz},
	\begin{align*}
		\bbE_{}\left[p_0^*\pST (\stt_0^*,p^*_0) \gets \valest_{V,\bB}(\stt_0,1^{2/\eta})\right]&\\ &= \Pr\left[V(r,y) = 1\pST \begin{array}{l}
			r\gets \zo^d  \\
			y \gets \bB(1^\secp,r,\stt_0) 
		\end{array}\right] \\ 
		&= \Pr\left[V(r,y) = 1\pST \begin{array}{l}
			r\gets \zo^d  \\
			x = G(1^\secp,r)\\
			y \gets B(1^\secp,x,\stt_0) 
		\end{array}\right]  = p\enspace.
	\end{align*}
	We now turn to prove the persistence property.
	\def \hR {\widehat{R}}
\def \hcRs{\widehat{\cR}^*}

	\paragraph{The Purified $\hR$.} We consider the unitary purifications $\set{\hR_{\secp,i,x}}_{\secp,i,x}$ of $R_B(1^\secp,1^i,x,\cdot)$  that act on registers $SY\hY$. Each such unitary $\hR_{\secp,i,x}$ is composed of the unitary purifications of each of the steps $R_B(1^\secp,1^i,x,\cdot)$, which are performed coherently. We assume w.l.o.g that $Y\hY$ consist of the purifying registers for each one of the steps. In particular, $\hY$ includes registers $P_{i-1}P_{i}^*$ corresponding to the measurements $p_{i-1},p_i^*$ done by $\valest$, and $Y$ corresponds to the measurement of $y_i$ by $\Pi_{x_i}$. Throughout we rely on the fact that the distribution of $p_{i-1},p_i^*$ measured in any purified interaction $A_z^{\hcRs}$ is identical to their distribution in the non-purified interaction $A_z^{\cR^*}$.
	
		\begin{lemma}\label{lem:pers}
		Fix any $\stt^* \in \qsts{\cH}$ with purification $\hR_{\secp,0}(\qv{0},\qv{0})$. Then $\cR^* = (R_B,\stt_0^*)$ is $\frac{\eta}{2}$-persistent, with respect to the purifications $\set{\hR_{\secp,i,x}}$, with persistent value $p_0$, where $p_0$ the purifying measurement corresponding to the first application of (the purified) $\valest$.
	\end{lemma}
	
	\begin{lemma}\label{lem:p0}
		Except with probability $\eta/2$, over $(\stt_0^*,p_0^*)\gets S_B(\stt_0)$ and $(\qv{\sigma}_{0},p_{0}) \gets \valest_{V,\bB}(\stt_{0}^*,1^{1/\varepsilon_1})$, 
		\begin{align*}
			\abs{p_0 - p_0^*} &\leq \eta/2\enspace.
		\end{align*}
	\end{lemma}
	\noindent
	Indeed, combining the two lemmas it follows the distribution on $\cR^* = (R_B,\stt_0^*)$ induced by $(\stt_0^*,p_0^*)\gets S(\stt_0)$ is $(p^*_0,\eta)$-consistent. 
	
	\bigskip\noindent
	We now prove the above lemmas. Toward proving Lemma \ref{lem:pers}, we first prove two useful claims. In what follows let $\hvalest_\varepsilon$ be a unitary purification of $\valest_{V,\bB}(\cdot,1^{1/\varepsilon})$ that operates on registers $SPZ$, where given $(\qv{\rho},\qv{0},\qv{0})$, it outputs $(\qv{\sigma},\qv{p},\qv{z})$ where $(\qv{\sigma},\qv{p})$ have the same density matrix as $\valest(\qv{\rho},1^{1/\varepsilon})$. Then given any pure input state $\qv{\rho}$, measuring $(\qv{p},\qv{z})$ in the computational basis, purifies $\qv{\sigma}$. We denote by $(\qv{\sigma},p,z)\gets \hvalest(\cdot,1^{1/\varepsilon})$ the process that applies $\valest_\varepsilon$ and measures $(\qv{p},\qv{z})$.

	\begin{clm}[Value of Post-Estimation State]\label{clm:post_est_stt}
		Let $\qv{\rho}\in \qsts{\cH}$ be a pure state. Then with probability $1-\varepsilon$ over the measurement of  $(p,z)$ in $(\qv{\sigma},p,z) \gets \hvalest(\qv{\rho},1^{1/\varepsilon})$, it holds that
		$$
		\abs{p - \Pr\left[V(r,y) = 1\pST \begin{array}{l}
				r\gets \zo^d  \\
				x = G(1^\secp,r)\\
				y \gets \bB(1^\secp,r,\qv{\sigma}) 
			\end{array}\right]} \leq \varepsilon \enspace.
		$$
	\end{clm}
	\begin{proof}
		Consider applying $(\qv{\sigma},p,z) \gets \hvalest(\qv{\rho},1^{1/\varepsilon})$ and then applying $(\qv{\sigma}^{*},p^{*}) \gets \valest_{V,\bB}(\qv{\sigma},1^{1/\varepsilon})$. Then by the fact that estimation is almost projective (Lemma \ref{lem:csmz}), it holds with probability $1-\varepsilon$ that $|p^{*} - p| \leq \varepsilon$. Also,
		$$
		\Pr\left[V(r,y) = 1\pST \begin{array}{l}
			r\gets \zo^d  \\
			x = G(1^\secp,r)\\
			y \gets \bB(1^\secp,r,\qv{\sigma})
		\end{array}\right] = \bbE[p^{*}] = p + \bbE(p^{*}-p)\enspace.
		$$
		The claim follows.
	\end{proof}
	
	\begin{clm}[$\Pi_x$ vs $\bB$]\label{clm:proj_val}
		For any state $\qv{\rho}\in \qsts{\cH}$,
		\begin{align*}
			\Pr\left[V(r,y) = 1\pST \begin{array}{l}
				r\gets \zo^d  \\
				x = G(1^\secp,r)\\
				(\qv{\rho}^*,y) \gets \Pi_{x}(\qv{\rho})
			\end{array}\right] = \Pr\left[V(r,y) = 1\pST \begin{array}{l}
				r\gets \zo^d  \\
				x = G(1^\secp,r)\\
				y \gets \bB(1^\secp,r,\qv{\rho})
			\end{array}\right]\enspace.
		\end{align*}
	\end{clm}
	\begin{proof}
		By the definition of $\Pi_x$ it acts exactly as $\bB(1^\secp,r,\cdot)$ with the exception that $\Pi_x$ replaces with $\bot$ any $y \notin Y_x = \set{y:\imver(1^\secp,x,y)=1}$. Recall, however, that $Y_x$ contains all valid solutions $\set{y:V(1^\secp,r,y) = 1}$, and hence (assuming w.l.o.g that $V(1^\secp,r,\bot) \neq 1$), this difference does not affect whether $V(1^\secp,r,y) = 1$. 
	\end{proof}

	\noindent
	We now prove Lemma \ref{lem:pers}.
	\begin{proof}[Proof of Lemma \ref{lem:pers}.]
		Fix any solver-aided algorithm $A$ with input $z$. We consider the random variables $\stt_0^*,p_0,p_1^*,\stt_1^*,p_1,p_2^*,\stt_2^*,p_2,\dots$ given by extended the solver interaction $A_z^{\hcRs}$, where $\stt_i^*$ are the corresponding pure state and $p_{i-1},p_i^*$ are the purifying measurements of registers $P_{i-1}P_i^*$.
		
		\begin{clm}[The Many Shot Value]\label{clm:ms_value}
			For all $i\geq 0$, except with probability $\varepsilon_{i+1}$ over $A_z^{\hcRs}$,
			$$ \abs{p_i - \val_P[i,A_z^{\hcRs}] } \leq \varepsilon_{i+1}\enspace,
			$$
			where each random variable $\val_P[i,A_z^{\hcRs}]$ is determined by the random variable $\stt_i^*$.
		\end{clm}
		\begin{proof}
			For every $i$, 
			\begin{align*}
				\val_P[i,A_z^{\hcRs}] = &\Pr\left[V(r,y) = 1\pST \begin{array}{l}
				(\qv{\sigma}_{i},p_{i},z) \gets \hvalest(\stt_{i}^*,1^{1/\varepsilon_{i+1}})\\
					r\gets \zo^d  \\
					x = G(1^\secp,r)\\
					(\qv{\sigma}_{i+1}^*,y) \gets \Pi_{x}(\qv{\sigma}_{i})
				\end{array}\right] =\\ &\Pr\left[V(r,y) = 1\pST \begin{array}{l}
				(\qv{\sigma}_{i},p_{i},z) \gets \hvalest(\stt_{i}^*,1^{1/\varepsilon_{i+1}})\\
					r\gets \zo^d  \\
					x = G(1^\secp,r)\\
					y \gets \bB(1^\secp,r,\qv{\sigma}_{i})
				\end{array}\right]\enspace,
			\end{align*}
			where the first equality follows by the definition of $\val_P[i,A_z^{\hcRs}]$ and the second by Claim \ref{clm:proj_val}. The claim now follows from Claim \ref{clm:post_est_stt}.
		\end{proof}
		
		\begin{clm}[Persistence]\label{clm:persistence} 
			For all $i\geq 1$, except with probability $ 2 \varepsilon_i$ over $A_z^{\hcRs}$, 
			\begin{align*}
				\abs{p_{i} - p_{i}^*} &\leq \varepsilon_{i}\enspace,\\
				\abs{p_{i}^*-p_{i-1}} &\leq \varepsilon_{i}\enspace.
			\end{align*}
			
		\end{clm}
		\begin{proof}
			The first inequality follows from the estimation is almost projective guarantee and second from the repairing guarantee (both given by Lemma \ref{lem:csmz}).
		\end{proof}
		Combining Claims \ref{clm:ms_value} and~\ref{clm:persistence}, and applying a union bound, we deduce that $3\sum_i\varepsilon_i$-persistence holds, with persistent value $p_0$. The Lemma now follows by our choice of $\varepsilon_i$:
		$$
		3\sum_i\varepsilon_i = \frac{3\eta}{\pi^2} \sum_i i^{-2} = \eta/2\enspace.
		$$
	\end{proof}%

	\begin{proof}[Proof of Lemma \ref{lem:p0}.]
		The lemma follows from the estimation is almost projective guarantee (Lemma \ref{lem:csmz}).
	\end{proof}
	This concludes the proof of Theorem \ref{thm:state_restoration}.
\end{proof}

\section{Stateful Solvers To Memoryless Solvers}
\label{sec:memoryless}

The following theorem shows that it is possible to convert stateful solvers into memoryless solvers with the same value, albeit with a few caveats. First, the distribution of queries that is to be made to the memoryless solver needs to be known ahead of time (i.e.\ it needs to be decided upfront in a non-adaptive manner). Second, the resulting memoryless solver might not be efficiently executable. Instead, we provide a simulator that can emulate its behavior, but only once, and only on an input that comes from the prescribed distribution. The simulator only manages to simulate the execution up to some statistical error, and its running time is polynomial in the inverse of this error. A formal theorem statement follows.

\begin{theorem}\label{thm:simmemless}
	There exists a polynomial time oracle-aided simulator $\Sim$ with the following properties.
	Let $\cB$ be a $(p, \eta)$-persistent $\ell$-stateful solver for a falsifiable non-interactive assumption $P$  and let $D=\{D_\secp\}_\secp$ be an efficiently samplable distribution ensemble over $k$-tuples of $P$ instances. Finally, let $\delta$ be some parameter. Then there exists a $(p, \eta)$-persistent (but possibly inefficient) distribution over memoryless solvers $\cB'=\cB'_{\ell, D, \delta} = (B', \emptyset)$ for $P$ such that the following holds.  
	
	Consider sampling $\vec{x}\gets D_\secp$, and let $\cB'(1^\secp, \vec{x})$ be the transcript of the process that feeds the elements of $\vec{x}$ into $\cB'$ one-by-one in order (i.e.\ executes $B'(1^\secp, 1^i, x_i, \emptyset)$ in order).
	Then $\Sim^{\cB,D}(1^\secp, 1^\ell, 1^{1/\delta}, \vec{x})$ makes non-adaptive black-box access to $\cB$ and produces a distribution that is within at most $\delta$ statistical distance from $\cB'(1^\secp, \vec{x})$.
\end{theorem}
We note that our simulator is ``almost'' a black-box algorithm in $\cB$ in the sense that it takes the size of the state $1^\ell$ as input, but otherwise it only makes black-box queries to $\cB$. We also emphasize that the simulator does not depend at all on $p, \eta$ or any other property of $\cB$ (other than $\ell$).

\subsection{The Simulator $\Sim$}
\label{sec:memlesssim}

We start by describing the simulator that will be used to prove Theorem~\ref{thm:simmemless}. The simulator $\Sim$ simply ``floods'' the solver $\cB$ with queries from a fixed distribution, and plants the elements of $\vec{x}$ in random positions.

Specifically, $\Sim^{\cB,D}(1^\secp, 1^\ell, 1^{1/\delta}, \vec{x})$ works as follows. Let $t$ be such that $k\sqrt{\ell/2t} \le \delta$, i.e.\ $t = O(\ell (k/\delta)^2)$. The simulator is also going to generate a non-adaptive sequence of queries. We start by defining our ``flooding'' distribution.

\begin{definition}[Random Marginal]
	Let $D$ be a distribution over $X^k$, i.e.\ $k$-tuples over a domain $X$. Then the random marginal distribution $D_U$ over $X$ is a distribution obtained by sampling $(x_1, \ldots, x_k)$ according to $D$, sampling a random $i$ in $[k]$, and outputting $x_i$ as the final sample.
\end{definition}

The simulator starts by sampling the following values.
\begin{enumerate}
	\item A vector $\vec{z}$ of $k \cdot t$  samples $z_{j,i} \gets D_U$, where $j \in [k]$, $i \in [t]$.
	\item $k$ uniform samples $i_j \gets [t]$, where $j \in [k]$.
	\item A uniform permutation $\pi$ over $[k]$.
\end{enumerate}

It then generates a sequence of queries $\vec{z}^*$ by taking the vector $\vec{z}$ and, for all $j \in [k]$,  replacing $z_{j,i_j}$ with $x_{\pi(j)}$. Namely, thinking of $\vec{z}$ as containing $k$ sequences of queries of length $t$ each, we plug in a random element from $\vec{x}$ in a random location in each sequence.

The simulator then calls $\cB$ on the queries in $\vec{z}^*$ in order, to obtain a sequence of responses $\vec{y}$. Let $y_{j,i}$ be the $(j,i)$ element in this sequence. We define $y^*_j = y_{j, i_j}$. The simulator returns the transcript $((x_1, y^*_{\pi^{-1}(1)}), \ldots, (x_k, y^*_{\pi^{-1}(k)}))$. Namely, we output a transcript that pairs each $x_i$ with the response that $\cB$ produces when introduced to the query $z_{j,i_j}=x_{i}$, namely $\pi(j)=i$.

\subsection{Proving Theorem~\ref{thm:simmemless}}

We now turn to prove the theorem. We start by defining a hybrid distribution which is defined with respect to purifying executions of $\cB$. This will allow us to make claims about extended transcripts, and finally to redact to standard transcript and derive the proof of the theorem.

\paragraph{A Hybrid Distribution.}	%
To prove the theorem, we define the hybrid distribution $\cS_h$, defined for every $h\in\{0,1,\ldots, k\}$.
\begin{enumerate}
	\item Sample a uniform permutation $\pi$ over $[k]$. %
	\item For all $j \in [k]$, sample a random index $i_j \in [t]$. %
	\item Sample $\vec{x}$ from $D$. %
	\item Generate a sequence of queries $z_{j,i}$ for all $j\in [k]$, $i \in [t]$ as follows.
	\begin{enumerate}
		\item For all $j > h$, set $z_{j,i_j}=x_{\pi(j)}$.
		\item Otherwise sample $z_{j,i_j}$ from $D_U$.
	\end{enumerate}
	\item Generate the extended transcript $\hts$ of executing $\cB$ (in a purifying manner) on the entries $z_{j,i}$ in lexicographic order (i.e.\ starting with $(1,1), \ldots, (1,t)$ and concluding with $(k,1), \ldots, (k,t)$). We let $\hts_{j,i}$ denote the prefix of the transcript prior to making the $(j,i)$ query. We let $\ket{s_{j,i}}$ denote the solver state respective to $\hts_{j,i}$, as guaranteed by Proposition~\ref{prop:exttsstate}. Notice that $\ket{s_{1,1}}$ is the initial state $\stt_0$ of $\cB$ conditioned on $\hts_0 = \hy_0$.

	\item The output of the hybrid $\cS_h$ then consists the following values, for all $j \in [k]$:
	\begin{enumerate}
		\item The values $i_j, \pi(j)$.
		\item The quantum state in the beginning of the $j$-th run: $\ket{s_{j,1}}$.
		\item The quantum state right before the $i_j$-th query in the $j$-th sequence is made:  $\ket{s_{j, i_j}}$. 
		\item The value $x_{\pi(j)}$, which is the $i_j$-th query in the $j$-th sequence if $j > h$.  
		\item An answer $(y_{\pi(j)}, \hy_{\pi(j)})$ computed as follows. 
		\begin{itemize}
		\item If $j > h$ then set $(y_{\pi(j)}, \hy_{\pi(j)})= (y_{j, i_j}, \hy_{j, i_j})$ (i.e.\ the $(y,\hy)$-part of the $(j, i_j)$-th triple in $\hts$). 
		\item Otherwise generate $(y_{\pi(j)}, \hy_{\pi(j)})$ as $\ovb(1^\secp, 1^{t (j-1) + i_j}, x_{\pi(j)}, \ket{s_{j, i_j-1}})_{\mathrm{y, \hy}}$.
		\end{itemize}
	\end{enumerate}
\end{enumerate}

In what follows, we will prove that the distributions induced by the first and last hybrids are close in trace distance, as formalized below.

\begin{lemma}\label{lem:memlesshybridtotal}
	It holds that $\td(\cS_0, \cS_{k}) \le k\sqrt{\ell/(2t)}$.
\end{lemma}

Before proving Lemma~\ref{lem:memlesshybridtotal}, we argue that it implies the validity of Theorem~\ref{thm:simmemless}. Indeed, we observe that the output of the simulator $\Sim$ can be extracted from $\cS_0$ by simply outputting all of the pairs $((x_1, y_1), \ldots, (x_k, y_k))$.
Applying the same extraction procedure on the last hybrid $\cS_k$ will lead to a sequence $((x_1, y_1), \ldots, (x_k, y_k))$ in which $y_{\pi(j)} = B(1^\secp, 1^{t (j-1) + i_j}, x_{\pi(j)}, \ket{s_{j, i_j-1}})_\mathrm{y}$. However, in the hybrid $\cS_k$, the transcript $\hts$, and therefore all states $\ket{s_{j,i}}$, are generated independently of $\vec{x}$. Therefore, for every values of $\pi, \hts$ one could define a memoryless adversary $\cB' = (B'_{\pi, \hts}, \emptyset)$, defined by 
\begin{align}
	B'_{\pi, \hts}(1^\secp, 1^j, x, \emptyset) = B(1^\secp, 1^{t (j'-1) + i_{j'}}, x, \ket{s_{j', i_{j'}-1}})_\mathrm{y}~,
	\end{align} 
with $j' = \pi^{-1}(j)$. Note that the sequence of states is hard-wired into $B'$ and it does not require to propagate a state throughout the execution.  %

We therefore indeed have that the solver $\cB'$ is a distribution over memoryless solvers indicated by sampling $\pi, \hts$ from their respective distributions and executing $B'_{\pi, \hts}$. Since $\cB$ is $(p,\eta)$-persistent, we have that with probability $1-\eta$ over $\hts$, all invocations of $B(1^\secp, 1^{t (j'-1) + i_{j'}}, x, \ket{s_{j', i_{j'}-1}})$ have value $p \pm \eta$, which would imply that $(B'_{\pi, \hts}, \emptyset)$ is $(p,\eta)$-persistent. Therefore, the distribution $\cB'$ is also, by definition, $(p,\eta)$-persistent.

\medskip\noindent
The proof of Lemma~\ref{lem:memlesshybridtotal} will follow from a standard hybrid argument, given by the following lemma.
\begin{lemma}\label{lem:memlesshybrid}
	For all $h \in \{0,1, \ldots, k-1\}$ it holds that
	\begin{align} 
		\td(\cS_h, \cS_{h+1}) \le \sqrt{\ell/(2t)}~.
	\end{align}
\end{lemma}
\begin{proof}
	We will show that the lemma holds true even when conditioning both $\cS_h, \cS_{h+1}$ on any value for $\hts_{h,1}$ (the $(h \cdot t)$-prefix of the transcript $\hts$). 
	
	We will show that the lemma follows from the following claim.
	\begin{clm}\label{claim:memlesszoomin}
		Conditioning on any value of $\hts_{h,1}$ for both $\cS_h, \cS_{h+1}$, the joint distribution of: 
		\begin{align}
			(i_h, \hts_{h, i_h}, \ket{s_{h+1,1}}, (x_{\pi(h)}, y_{\pi(h)},\hy_{\pi(h)}))
		\end{align}
		is within trace distance $\sqrt{\ell/(2t)}$ between $\cS_h, \cS_{h+1}$.
	\end{clm}
	
	Given Claim~\ref{claim:memlesszoomin}, Lemma~\ref{lem:memlesshybrid} follows since all other elements of the two distributions $\cS_h, \cS_{h+1}$ can be sampled given $\hts_{h,1}$ and $(i_h, \hts_{h, i_h}, \ket{s_{h+1,1}}, (x_{\pi(h)}, y_{\pi(h)}, \hy_{\pi(h)}))$, as follows.
	\begin{enumerate}
		\item Sample the permutation $\pi$ and the query vector $\vec{x}$ conditioned on the value $x_{\pi(h)}$.
		\item For very $j\in[k]\setminus \{h\}$, sample $i_j$ uniformly in $[t]$.
		\item For all $j < h$, the transcript prefix $\hts_{h,1}$ determines all states $\ket{s_{j,i}}$ (for all $i\in[t]$), which in turn, together with $\vec{x}$, determines the distribution of $y_{\pi(j)}, \hy_{\pi(j)}$ for all $j < h$ (since this distribution is specified by applying the solver $B$ on $x_{\pi(j)}$ with quantum state that is determined by the $h$-prefix).
		\item For all $j > h$ the outputs of both $\cS_h, \cS_{h+1}$ are determined as the outcomes of an identical quantum process applied to the state $\ket{s_{h+1, 1}}$ (the initial state of the $(h+1)$-th sequence), considering that $\pi$ and $\vec{x}$ have been determined.
	\end{enumerate}

	We now proceed to prove Claim~\ref{claim:memlesszoomin}, and focus on the distribution of $(i_h, \hts_{h, i_h}, \ket{s_{h+1,1}}, (x_{\pi(h)}, y_{\pi(h)}, \hy_{\pi(h)}))$ in the two hybrids, given that $\hts_{h,1}$ is fixed. The claim follows straightforwardly from our information theoretic Plug-In Lemma (Lemma~\ref{lem:plugin}), where the classical values $y_i$ in the lemma corresponds to pairs $(z_{h,i}, y_{h,i}, \hy_{h,i})$ generated in the $h$'th round in the hybrid experiment.  Note that since we fixed $\hts_{h,1}$, the distribution over these classical values is also fixed, and indeed the value $\qv{s} = \ket{s_{h+1, 1}}$ depends on this sequence of $t$ values. The triple $(x_{\pi(h)}, y_{\pi(h)}, \hy_{\pi(h)})$ differs between $\cS_h$ and $\cS_{h+1}$ since in the former it is exactly equal to the $i_{h+1}$ element in the $h$-th sequence, and in the latter it is sampled from the marginal distribution of this element. We can therefore apply the plug-in lemma directly to obtain the $\sqrt{\ell/(2t)}$ bound on the trace distance as Claim~\ref{claim:memlesszoomin} requires. This completes the proof of the claim and thus also of the lemma.
\end{proof}

\section{Memoryless Solvers To Stateless Solvers}
\label{sec:memlesstostateless}

\begin{theorem}\label{thm:simstateless}
	There exists a polynomial-time oracle-aided simulator $\simsl$ with the following properties.
	Let $\cB$ be a $(p, \eta)$-persistent memoryless solver for a falsifiable non-interactive assumption $P$ and let $\{D_\secp\}_\secp$ be an efficiently samplable distribution ensemble over $k$-tuples of $P$ instances. Let $\delta$ be some parameter. 
	
	Then there exists a $(p, \eta)$-persistent (but possibly inefficient) stateless solver $\cB''=\cB''_{\delta} = (B'', \emptyset)$ for $P$ such that the following holds. Consider sampling $\vec{x}\gets D_\secp$, and let $\cB''(1^\secp, \vec{x})$ be the transcript of the process that feeds the elements of $\vec{x}$ into $\cB''$ (i.e.\ executes $B''(1^\secp, x_i, \emptyset)$ for all $x_i$).
	Then $\simsl^{\cB}(1^\secp, 1^{1/\delta}, \vec{x})$ makes non-adaptive black-box access to $\cB$ and produces a distribution that is within at most $\delta$ statistical distance from $\cB''(1^\secp, \vec{x})$.
\end{theorem}
\begin{proof}
	The simulator $\simsl^{\cB}$ runs as follows. Given $\vec{x}$ as input, it generates a query vector $\vec{x}'$ of length $t=k^2$ as follows. It samples, without repetitions, $k$ indices $i_1, \ldots, i_k$ and sets $x'_j = x_{i_j}$. All other values of $x'$ are set to $0$ (or some other fixed value).
	
	After making the queries in $\vec{z}$ to $\cB$ and receiving an output vector $\vec{y}'$, the simulator sets $y_j = y_{i_j}$ returns $((x_1, y_1), \ldots, (x_k, y_k))$.
	
	Let us now define the stateless adversary $\cB''$. On input $x$, $B''(1^\secp, x)$ samples $j \gets [t]$ uniformly, and outputs $y = B(1^\secp, 1^j, x, \emptyset)_y$. The solver $\cB''$ is also $(p, \eta)$-persistent; indeed, its value is the average of values, which are all $\eta$-close to $p$. (Recall Remark~\ref{rmk:persistentmemless} about persistent values for stateless and memoryless solvers.)
	
	To bound the statistical distance between $\simsl^{\cB}(1^\secp, 1^{1/\delta}, \vec{x})$ and $\cB''(1^\secp, \vec{x})$, we consider the case where in the course of the execution of $\cB''(1^\secp, \vec{x})$, all $j$'s that are sampled are distinct. This happens with probability at least $1-k^2/t=1-\delta$. Conditioned on this event, $\cB''(1^\secp, \vec{x})$ is identically distributed as $\simsl^{\cB}(1^\secp, 1^{1/\delta}, \vec{x})$. It follows that in general the statistical distance is bounded by $\delta$.
\end{proof}

We conclude with a corollary that combines Theorem~\ref{thm:simmemless} and Theorem~\ref{thm:simstateless}. %
\begin{corollary}\label{cor:statefulstateless}
	There exists a polynomial-time simulator $\simall$ with the following properties.
		Let $\cB$ be a $(p, \eta)$-persistent $\ell$-stateful solver for a falsifiable non-interactive assumption $P$ and let $\{D_\secp\}_\secp$ be an efficiently samplable distribution ensemble over $k$-tuples of $P$ instances. Finally, let $\delta$ be some parameter. 
	
		Then there exists a $(p, \eta)$-persistent (but possibly inefficient) distribution over stateless solvers $\cB''=\cB''_{\ell, D, \delta} = (B'', \emptyset)$ for $P$. Consider sampling $\vec{x}^*\gets D_\secp$, and let $\cB''(1^\secp, \vec{x}^*)$ be the transcript of the process that feeds the elements of $\vec{x}^*$ into $\cB''$ (i.e.\ executes $B''(1^\secp, x^*_i, \emptyset)$ for all $x^*_i$). Then $\simall^{\cB,D}(1^\secp, 1^\ell, 1^{1/\delta}, \vec{x}^*)$ makes non-adaptive black-box access to $\cB$ and produces a distribution that is within at most $\delta$ statistical distance from $\cB''(1^\secp, \vec{x}^*)$.
\end{corollary}
\begin{proof}
The simulator $\simall^{\cB,D}(1^\secp, 1^\ell, 1^{1/\delta}, \vec{x}^*)$ runs as follows. Set $\delta'=\delta/2$.

\begin{enumerate}
	
	\item Define an efficiently samplable distribution $D'$ over sequences of $P$ instances as follows. Consider the non-adaptive black-box simulator $\simsl$ from Theorem~\ref{thm:simstateless}. Start by sampling $\vec{x} \gets D$. Run $\simsl^{(\cdot)}(1^\secp, 1^{1/\delta'}, \vec{x})$ up until the point where it generates its sequence of oracle queries $\vec{x}'$ (note that to this end there is no need to actually have any access to the solver itself). Let $\vec{x}'$ be the sample of $D'$.

	\item\label{i:startsimsl} Start an execution $\simsl^{(\cdot)}(1^\secp, 1^{1/\delta'}, \vec{x}^*)$, until the point where the sequence of queries $\vec{x}'^*$ is generated. 
	
	\item Consider the non-adaptive black-box simulator $\Sim$ from Theorem~\ref{thm:simmemless}. Execute the simulator $\Sim^{\cB, D'}(1^\secp, 1^\ell, 1^{1/\delta'}, \vec{x}'^*)$ to obtain a transcript $\ts$.
	
	\item Resume the execution of $\simsl$ from step~\ref{i:startsimsl}, plugging in the responses from $\ts$ as the solver outcome. Produce the output of $\simsl$ as the output of $\simall$.	
\end{enumerate}

To analyze, we first note that by definition $\vec{x}'^*$ is sampled from the distribution $D'$. Theorem~\ref{thm:simmemless} implies that there exists a distribution $\cB'$ over memoryless adversaries such that the transcript $\ts$ is withing $\delta'$ statistical distance from having been produced by $\cB'(1^\secp, \vec{x}'^*)$. It therefore follows that the output of $\simall^{\cB,D}(1^\secp, 1^\ell, 1^{1/\delta}, \vec{x}^*)$ is within statistical distance $\delta'$ from $\simsl^{\cB'}(1^\secp, 1^{1/\delta'}, \vec{x}^*)$. We can now apply Theorem~\ref{thm:simstateless} to deduce that for each memoryless solver $\cB'_0$ in the support of $\cB'$, the latter is within $\delta'$ statistical distance from some $\cB''_0(1^\secp, \vec{x}')$ where $\cB''_0$ is stateless. This therefore induces a distribution over stateless solvers. Applying the union bound we get that $\simall^{\cB,D}(1^\secp, 1^\ell, 1^{1/\delta}, \vec{x}^*)$ is within statistical distance at most $2\delta'=\delta$ from $\cB''(1^\secp, \vec{x}')$ as required.

By definition, with probability $1-\eta$ over the sampling of the memoryless solver $\cB'_0$ from the distribution, $\cB'_0$ itself is $(p,\eta)$-persistent, this property carries over to $\cB''_0$. Therefore, the distribution $\cB''$ is by definition $(p,\eta)$-persistent.
\end{proof}

\section{Classical Non-Adaptive Reductions and Quantum Solvers}
\label{sec:classicalna}

In this section, we show that a wide class of classical reductions can be translated to the quantum setting. Specifically we start from any non-adaptive black-box reductions from classically solving $P$ with a verifiably-polynomial image (Definition~\ref{def:ver_poly_im}), to classically solving $Q$. We transform it into a quantum reduction from quantumly solving $P$ to quantumly solving $Q$.

\begin{theorem}\label{thm:non_adaptive}
Assume there exists a classical non-adaptive black-box reduction from solving a non-interactive assumption $Q$ to solving a non-interactive assumption $P$ with a verifiably-polynomial image. Then there exists a quantum reduction from solving $Q$ to quantumly solving $P$. This reduction is durable if the original classical reduction has positive advantage.
\end{theorem}

\newcommand{\smthns}{\mathsf{s}}

\begin{proof}
Let $\cR$ be a classical non-adaptive black-box reduction from solving a non-interactive assumption $Q=(G_Q, V_Q, c_Q)$ to solving a non-interactive assumption $P=(G_P, V_P, c_P)$. We present a quantum reduction $\cR'$ from solving $Q$ to quantumly solving $P$. We start by describing and analyzing $\cR'$ with a one-shot advantage, and then extend it to address durability in the case that $\cR$ has positive advantage. We assume w.l.o.g that $\cR$ never makes the same query twice to its oracle function (see Remark \ref{rem:clas_red_wlog}). 

Recalling Definition~\ref{def:reduction}, $\cR'$ takes as input $(1^\secp, 1^t, x_Q, \stt)$, where $x_Q \in \binset^{n_Q}$ is potentially an instance of $Q$, and its initial state is $\stt'_0 = (\stt_{0},B, 1^{1/\varepsilon},1^{T})$, where we are guaranteed that $\cB=(\stt_0, B)$ is a $P$ solver with advantage at least $\varepsilon$ that runs in time at most $T$. 

We let $\varepsilon'$ denote the advantage of $\cR$ in solving $Q$ when given access to an oracle that solves $P$ with advantage at least $\varepsilon/2$. We are guaranteed that $\varepsilon' = \poly(\varepsilon,\secp^{-1})$. We set $\delta = \varepsilon'/2$ and $\eta = \min\{\varepsilon/4, \varepsilon'/2 \}$.

We define a distribution $D$ over $(\binset^{n_P})^k$ as the distribution over the set of oracle queries produced by first sampling a uniform $r'_Q$ and using it to generate $x'_Q=G_Q(1^\secp, r'_Q)$, and finally executing $\cR(1^\secp,1^{4/\varepsilon}, x'_Q)$ to produce a $k$-tuple of $P$-instances.

Having all of these definitions in place, we can now introduce the execution of $\cR'(1^\secp, 1^0, x_Q, \stt'_0)$. Namely, we start by analyzing the one-shot execution of $\cR'$ (the case $t=0$).%
\begin{enumerate}
	\item\label{i:restore} Let $R,S$ be the state restoration algorithms with respect to $P$ as guaranteed by Theorem~\ref{thm:state_restoration}, with parameter $\eta$ as defined above. Set $(\stt_0^*,p^*) \gets S_B(\stt_0)$. Define $\cB_0 = \cR^* = (R_B, \stt^*_0)$ and recall that $\cB_0$ is $(p^*, \eta)$-persistent, and that $\Ex[p^*]=p$.
	
	\item\label{i:startr} Execute $\cR(1^\secp,1^{3/\varepsilon}, x_Q)$ to obtain the sequence of queries $\vec{x}$.
	
	\item Recall the simulator $\simall$ guaranteed by Corollary~\ref{cor:statefulstateless}. Execute $\simall^{\cB_0,D}(1^\secp, 1^\ell, 1^{1/\delta}, \vec{x})$ to obtain a transcript $\ts$.
	
	\item\label{i:extractresp} Extract the responses to $\vec{x}$ from $\ts$ and resume the execution $\cR$ from step~\ref{i:startr} with these responses. Once the execution of $\cR$ completes and a value $y_Q$ is output, output $y_Q$ as the output of $\cR'$. 
\end{enumerate}

To analyze the one-shot value and advantage of $\cR'$, we start by analyzing the performance of $\cR'$ conditioned on obtaining a fixed value $p^*$ in step~\ref{i:restore} of the execution. In this case $\cB_0$ is $(p^*,\eta)$-persistent, and we can invoke Corollary~\ref{cor:statefulstateless} to conclude that there exists a $(p^*,\eta)$-persistent distribution over stateless adversaries $\cB_{p^*}''$ s.t.\ the output of $\cR'$ is within statistical distance $\delta$ from the execution of $\cR^{\cB_{p^*}''}(1^\secp,1^{4/\varepsilon},
 x_Q)$.

In turn, the execution of $\cR^{\cB_{p^*}''}(1^\secp, 1^0, x_Q, \stt'_0)$ is equivalent to executing $\cR^{\cB''}(1^\secp,1^{4/\varepsilon},
 x_Q)$, where $\cB''$ is a distribution over stateless solvers defined as follows.  First sample $p^*$ from its designated distribution, then sample $\cB_{p^*}''$ from the $(p^*,\eta)$-persistent distribution of stateless solvers. Recall that with probability $1-\eta$ over the sampling of $\cB_{p^*}''$, it holds that the outcome is a (single) $(p^*,\eta)$-persistent stateless solver and therefore that $\abs{\val_P[0, \cB_{p^*}''\big] - p^*} \le \eta$. It follows that with probability at least $1-\eta$:
\begin{align*}
	\abs{\Ex[\val_P[0, \cB'']]-p} &= \abs{ \Ex\big[\val_P[0, \cB_{p^*}''\big] - p^*\big]}\\
	& \le \Ex\big[\abs{\val_P[0, \cB_{p^*}''\big] - p^*}\big]\\
	& \le \eta~.
\end{align*}
It follows that $\cB''$ has advantage at least $\varepsilon-2\eta \ge \varepsilon/2$ in solving $P$.
We have therefore that $\cR^{\cB''}$ has advantage at least $\varepsilon'$ in solving $Q$. Since the output of $\cR'$ is within $\delta = \varepsilon'/2$ statistical distance from $\cR^{\cB''}$, we conclude that $\cR'$ has advantage at least $\varepsilon'/2$. We therefore established the one-shot value of $\cR'$. 

It remains to extend the definition of $\cR'$ beyond $t=0$ in order to establish that it is durable when $\cR$ has positive advantage. The basic idea is to propagate the final state of $\cB_0$ at the end of step \ref{i:extractresp} as the initial state of the next execution, and use the persistence of $\cB_0$ in order to execute steps \ref{i:startr}-\ref{i:extractresp} anew for each input.

In order to formalize the above intuition, we require the following definitions. We define a ``shifted execution'' of a solver as follows. Letting $\cB = (B, \stt_0)$ be a solver. We define the solver $\cB_{+j} = (B_{+j}, \stt_0)$ via $B(1^\secp, 1^t, x, \stt) = B(1^\secp, 1^{t+j}, x, \stt)$. Namely, $\cB_{+j}$ simply executes $\cB$ but with a fixed offset in the $t$ input. A second notation that we require is for the maximal number of $\cB_0$ calls that are made in steps \ref{i:startr}-\ref{i:extractresp} of the execution above. We denote this value by $M$ and note that it is w.l.o.g a polynomial in $\secp,\varepsilon^{-1}$ that does not depend on the value of $x_Q$. %

Our durable reduction is therefore as follows: 
\begin{itemize}
	
	\item We extend the execution of $\cR'$ for $t=0$ defined above as follows. First, we ensure that steps \ref{i:startr}-\ref{i:extractresp} make \emph{exactly} $M$ queries to $\cB_0$, by inserting dummy queries if needed. Second, we specify the output state of the execution to be the output state of $\cB_0$ after the last call that has been made.
	
	\item For a value of $t > 0$ the execution of $\cR'(1^\secp, 1^t, x, \stt)$ is by executing steps \ref{i:startr}-\ref{i:extractresp} above (with the padding to $M$ queries), but using the shifted solver $\cB_{0+tM}$ instead of $\cB_0$. The output state is again the final output state of the solver $\cB_{0+tM}$.
\end{itemize}

\def \hcR {\widehat{\cR}}
Note that $\cB_0$ is $(p^*,\eta)$-persistent with respect to some purification $\hcB_0$. We can consider a corresponding purification $\hcR'$ of $\cR'$. We note that in an extended interaction $A_z^{\hcR'}$, letting $\varepsilon^* = |p^*-c_P|$, it holds that with probability $1-\eta$, for every $i$: $$\val_Q[i,A_z^{\hcR'}] \geq c_Q+\varepsilon' \enspace,$$
where $\varepsilon'\geq 0$ and if $\varepsilon^*-\eta \geq \varepsilon/4$, then $ \varepsilon' \geq \poly(\varepsilon^*-\eta,\secp^{-1})$.

Indeed, since $\cB_0$ is $(p^*,\eta)$-persistent, in the $i$-th underlying invocation of $\cB_{p^*}''$, its value as a $P$-solver is $\eta$-close to $p^*$, and hence its advantage in the corresponding invocation of $\cR(x_Q,1^\secp,1^{4/\varepsilon})$ is at least $\varepsilon^* - \eta$. If $\varepsilon^* - \eta \geq \varepsilon/4$, $\cR$ is guaranteed to have positive advantage $\poly(\varepsilon^*-\eta,\secp^{-1})$, in which case the corresponding value is $c_Q+\poly(\varepsilon^*-\eta,\secp^{-1})$.

It is left to argue that
$$
\bb{E}[\varepsilon'] \geq \poly(\varepsilon,\secp^{-1})\enspace.
$$
This follows from an averaging argument. With probability at least $\varepsilon/2$, $\varepsilon^* \geq \varepsilon/2$. In particular with probability at least $\varepsilon/2 -\eta \geq \varepsilon/4$ it holds that $\varepsilon^*- \eta \geq \varepsilon/4$, in which case $\varepsilon' \geq \poly(\varepsilon/4,\secp^{-1})$. 
\end{proof}

\section{An Impossibility Result for Search Assumptions}\label{sec:search}

Our result in Section \ref{sec:classicalna} transforms a classical non-adaptive reduction $\cR$ from solving $Q$ to classically solving $P$ into a reduction $\cR$ to {\em quantumly} solving $P$. It is restricted to assumptions $P$ with a verifiably-polynomial image. While this captures a large class of assumptions, such as all decision assumptions, it certainly does not capture all assumptions of interest. In particular, it does not capture {\em search assumptions} where the number of possible solutions per instance could be super polynomial, such as say the hardness of inverting a one-way function where the preimage size could be super-polynomial. 

In this section we show that this is somewhat inherent. We prove that for search assumptions, such a transformation cannot exist as long as the resulting reduction $\cR'$ is explicit in the assumptions $P,Q$. In particular, it may obtain as input the code of the algorithms describing $P,Q$, but does not get any implicit non-uniform advice regarding these assumptions. Indeed, the transformation in \ref{sec:classicalna} as well as the Persistence Theorem \ref{thm:state_restoration} on which it relies, the resulting quantum reduction $\cR'$ is in fact black-box in the assumptions $P,Q$, and in particular explicit.

\begin{definition}[Assumption Pair Colletion]
An assumption pair collection $\cP\cQ$ consists of pairs of assumptions $(P,Q)$, each given by its corresponding (possibly non-uniform) algorithms $(G_P,V_P,c_P)$ and $(G_Q,V_Q,c_Q)$.
\end{definition}

\begin{definition}[Explicit Reduction]
An explicit quantum reduction for assumption pair collection $\cP\cQ$ is an efficient algorithm $\cR$ with the following guarantee. For any $(P,Q)\in (\cP,\cQ)$ and any quantum solver $\cB_P = (B_P, \stt_0)$ for $P$ with one-shot advantage $\varepsilon$ and running time $T$, let $\stt'_0 = (\stt_{0},(P,Q),B_P, 1^{1/\varepsilon},1^{T})$. Then $\cB_Q = (\cR,\stt'_0)$ is a solver for $Q$ with one-shot advantage $ \poly(\varepsilon,T^{-1},\secp^{-1})$ and running-time $\poly(T,\varepsilon^{-1},\secp)$.

We say that the reduction is strongly explicit, instead of being given the explicit description of $(P,Q)$ as part of its input, it is given oracle access to its corresponding algorithms.

\end{definition}
Note that in the above definition $\stt'_0$ is formally a sequence 
$$\stt'_{0,_\secp} = (\stt_{0,\secp},(P,Q)_{\secp},B_{P,\secp}, 1^{1/\varepsilon(\secp)},1^{T(\secp)})\enspace,$$
where $(P,Q)_{\secp}$ consist of their corresponding algorithms (possibly along with their corresponding non-uniform advice) restricted to security parameter $\secp$ (w.l.o.g circuits).

Restating our result from Section \ref{sec:classicalna}, we proved that for any pair collection $\cP\cQ$, if for any $(P,Q)\in \cP,\cQ$, $P$ has verifiably-polynomial image, and there exists a classical non-adaptive black-box reduction $\cR_{P,Q}$ from solving $Q$ to solving $P$, then there also exists a strongly explicit quantum reduction $\cR'$ for $\cP\cQ$. We prove that if $P$ does not have a verifiably-polynomial image this may not be the case.

\begin{theorem}\label{thm:explicit_reduction}
There exists an assumption pair collection $\cP\cQ$, such that for any $(P,Q)\in \cP,\cQ$, there exists a classical non-adaptive black-box reduction $\cR_{P,Q}$ from solving $Q$ to solving $P$, but there is no strongly explicit reduction $\cR'$ for $\cP\cQ$. Assuming also post-quantum indistinguishability obfuscation, there also does not exist and explicit reduction $\cR'$.
\end{theorem}

We will restrict extension to ruling out explicit reductions based on indistinguishability obfuscation. The result for strongly explicit reductions is a direct extension.

\paragraph{The Collection $\cP\cQ$.} 
The collection is associated with a particular signature scheme $(\Gen,\Sig,\ver)$, with a corresponding message space $\cM=(\cM_\secp)_{\secp\in\mathbb{N}}$.  Each pair of assumptions $(P,Q)\in\cP\cQ$ has the following form:  
\begin{enumerate}
        \item For $i\in\{P,Q\}$, $P_i=(G_i,V_i,0)$, $G_i$ is a uniform generator and $V_i$ is a non-uniform verifier.
        
        \item $G_P$ takes as input $1^\secp$ and outputs a random message $x\leftarrow \cM_\secp$, whereas $G_Q$ takes as input $1^\secp$ and outputs two random and independent messages $x_1,x_2\leftarrow \cM_\secp$.
        \item $V_P$ and $V_Q$ have the same $\pk \in\Gen(1^\secp)$ hardwired into their description. 
        
        $V_P$ takes as input $(1^\secp,x,\sigma)$ and outputs~$1$ if and only if $\ver(\pk,x,\sigma)=1$, whereas $V_Q$ takes as input $(1^\secp,x_1,x_2,\sigma_1,\sigma_2)$ and it outputs $1$ if and only if $\ver(\pk,x_1,\sigma_1)=\ver(\pk,x_2,\sigma_2)=1$.
\end{enumerate}

\begin{claim}  For any signature scheme $(\Gen,\Sig,\ver)$ and corresponding collection $\cP\cQ$, there exists an efficient solver-aided algorithm ~$\cR$ such that for every $(P,Q)\in\cP\cQ$, $\cR$ is a classical non-adaptive black-box reduction from solving $P$ to solving $Q$.  
\end{claim}

\begin{proof}
We describe the reduction $\cR$.  On input $(1^\secp,(x_1,x_2))$, $\cR$ queries the solver with $x_1$ and $x_2$, and obtains $\sigma_1$ and~$\sigma_2$.  It outputs $(\sigma_1,\sigma_2)$. To complete the proof, we note that given any stateless classical $P$-solver $\cB$ with advantage $\varepsilon$, $\cR^\cB$ has advantage $\varepsilon^2$ in solving $Q$, as desired.  
\end{proof}

We now proceed to show that for an appropriately chosen signature scheme $(\Gen,\Sig,\ver)$ there is no explicit quantum reduction $\cR'$ for the collection $\cP\cQ$.

\paragraph{Tokenized Signature Schemes.} A tokenized signature scheme \cite{Ben-DavidS16} is a classical signature scheme $(\Gen,\Sig,\ver)$, with message space $\cM=\{\cM_\secp\}_{\secp\in\mathbb{N}}$, with two additional efficient quantum algorithms $(\TGen,\TSig)$. $\TGen$ takes as input a secret key $\sk$ and outputs a quantum state $\ket{\tk}$, referred to as a {\em signing token}, and $\TSig$ takes as input a signing token $\ket{\tk}$ and a message $m\in\cM$ and outputs a signature, with the guarantee that for every message $m\in\cM$, 
$$
\Pr[\ver(\pk,m,\sigma)=1]=1,
$$
where the probability is over $(\sk,\pk)\leftarrow\Gen(1^\secp)$, $\ket{\tk}\leftarrow \TGen(\sk)$ and  $\sigma\leftarrow\TSig(\ket{\tk},m)$.

\begin{definition}\label{def:secure-tk-sig}
A tokenized signature scheme $(\Gen,\Sig,\ver, \TGen,\allowbreak\TSig)$ is secure if for any efficient quantum adversary $\cA$ there exists a negligible function $\mu$ such that for every $\lambda\in\mathbb{N}$,
$$
\Pr\left[\begin{array}{c} m_1\neq m_2 ,\\
\ver(\pk,m_i,\sigma_i)=1~~\forall i\in[2]
\end{array}
\pST (m_1,m_2,\sigma_1,\sigma_2)\gets \cA(\pk,\ket{\tk})\right] = \mu(\lambda)\enspace,
$$
where the probability is over $(\sk,\pk)\leftarrow\Gen(1^\secp)$ and  $\ket{\tk}\leftarrow\TGen(\sk)$.
\end{definition}

We rely on the following result by Coladangelo et a..

\begin{theorem}[\cite{ColadangeloLLZ21}]
There exists a secure tokenized signature scheme assuming the existence of a post-quantum secure (classical) indistinguishability obfuscation scheme.
\end{theorem}

\begin{claim}
Let $(\Gen,\Sig,\ver,\TGen,\TSig)$ be a secure tokenized signature scheme, and let $\cP\cQ$ be the corresponding collection $\cP\cQ$. Then there exist no explicit quantum reduction for $\cP\cQ$.
\end{claim}

\begin{proof}
Assume toward contradiction that there exists an explicit reduction $\cR'$ for $\cP\cQ$, we show how an adversary $\cA$ can use it to $reak the secu$ity of the tokenized signatures (Definition \ref{def:secure-tk-sig}).

For security parameter $\secp$, $\cA$ is given a public key $\pk$ and token $\ket{\tk}$. $\cA$ samples $x_1,x_2 \gets \cM_\secp$ and invokes:
$$
\cR'((x_1,x_2),\stt'_0 = (\ket{\tk},(P,Q)_\secp,B_P, 1^{1/\varepsilon},1^{T})\enspace,
$$
where:
\begin{itemize}
    \item $(P,Q)_\secp$ are the circuits describing the assumption corresponding to $\pk$.
    \item
    $B_P$ is the quantum algorithm that given $(x,\ket{\tk})$, applies $\TSig$ to generate a signature $\sigma$ on $x$.
    \item
    $\varepsilon = 1$.
    \item
    and $T$ is the polynomial running time of $B_P$.
\end{itemize}
$\cA$ obtains back from the reduction two signatures $\sigma_1,\sigma_2$ and outputs $(x_1,\sigma_1,x_2,\sigma_2)$.

To see that $\cA$ breaks the security of the tokenized signature (namely manages to generate signatures on two different messages. Note that $B_P,\stt_0 = \ket{\tk}$ constitute a solver for $P$ (which generates one good signature) with probability $1$, accordingly the reduction $\cR'$ manages to solve $Q$ with noticeable probability, generating two signatures $\sigma_1,\sigma_2$.

\end{proof}

The above indeed relies on indistinguishability obfuscation to instantiate the tokenized signature scheme. To get an unconditional impossibility, but which only rules out {\em strongly} explicit reductions, we can rely on the following result of Ben-David and Sattath. 

\begin{theorem}[\cite{Ben-DavidS16}]
There exists a classical oracle distribution relative to which there exist (information theoretically) secure tokenized signature schemes.
\end{theorem}
The proof is a direct extension of the proof above.

\section{Proving the Plug-In Lemma}
\label{apx:pluginproof}

We recall the formal statement of the lemma.

\begin{lemma}[Lemma~\ref{lem:plugin}, restated]
	Let $\vec{Y} = (Y_1, \ldots, Y_t)$ be a joint distribution over $t$ classical random variables. Let $\vec{y}$ be distributed according to $\vec{Y}$. Let $\qv{s}$ be an $\ell$-qubit random variable that has arbitrary dependence on $\vec{y}$.	
We let $\vec{y}_i$ denote the prefix $\vec{y}_i = (y_1, \ldots, y_i)$ for $1 \le i \le t$, and $\vec{y}_0$ is the empty vector (and likewise for $\vec{Y}$).
Let $J$ be the uniform distribution over $[t]$ and let $j \gets J$. Define $y' \gets Y_J | (\vec{Y}_{j-1} = \vec{y}_{j-1})$.
Then it holds that
	\begin{align}
		\td( (j, \vec{y}_{j-1}, y_j, \qv{s}), (j, \vec{y}_{j-1}, {y}', \qv{s}) ) \le \sqrt{\ell/(2t)}~.
	\end{align}
\end{lemma}

We prove the lemma using tools from (quantum) information theory.
We recall the basic notions below and refer to \cite{QCQI,WatrousQIT} for additional reference. We let $H(\cdot)$ denote the entropy function both in the classical case (Shannon entropy) and in the quantum case (von Neumann entropy).

We denote the mutual information function by $I(\cdot:\cdot)$ both in the classical and in the quantum setting. 
We note that entropy or mutual information are only well defined for variables that have a well-defined joint density matrix.\footnote{Recall the following information-theoretic identities that hold both in the classical and quantum settings. Mutual information: $I(X:Y) = H(X)+H(Y)-H(X,Y)$. Conditional entropy: $H(X|Y) = H(X,Y)-H(Y)$. Conditional mutual information: $I(X:Y|Z) = H(X|Z)+H(Y|Z)-H(X,Y|Z) = H(X|Z) - H(X|Y,Z)$.}

Let $X$ and $Y$ be variables with a joint density matrix $\rho_{XY}$. We say that $Y$ is a classical random variable if its reduced density matrix is diagonal. If $Y$ is a classical variable then 
\begin{align}
	\rho_{XY} = \sum_{y} p_{y} \rho_{X|y} \otimes \ketbra{y}~.
\end{align}

In this case we refer to $\rho_{X|y}$ as the conditional density matrix of $X$ given $Y=y$, and refer to $X|y$ as the variable with this density matrix. In such a case it holds that 
$$H(X|Y) = \Ex_{y}[H(X|y)],
$$
where $y$ is distributed according to the classical distribution of $Y$.
We provide a proof for the sake of completeness.

\begin{proposition}\label{prop:classicalY}
	Let $X,Y$ be random variables with density matrix $\rho_{XY} = \sum_{y} \alpha_{y} \rho_{X|y} \otimes \ketbra{y}$. Then it holds that $H(X|Y) = \Ex_{y \sim Y}[H(X|y)]$.
\end{proposition}
\begin{proof}
	As stated in \cite[Theorem 11.8 (5)]{QCQI}, 
	it holds that $H(X,Y) = H(Y)+\Ex_{y \sim Y}[H(X|y)]$. Recalling that $H(X|Y) = H(X,Y) - H(Y)$, the proposition follows.
\end{proof}

\begin{corollary}\label{cor:posent}
	Letting $X,Y$ be as in Proposition~\ref{prop:classicalY} then $H(X|Y) \ge 0$.
\end{corollary}
\begin{proof}
	The corollary follows since for each $y$ it holds that $X|y$ is just a quantum variable and therefore it has non-negative entropy. The expectation over non-negative values remains non-negative.
\end{proof}

The following simple application of the chain-rule will be useful for us.
\begin{lemma}\label{lem:boundmutual}
	Let $\vec{y} = (y_1, \ldots, y_t)$ be a vector of arbitrarily distributed classical random variables, and let $\qv{s}$ be an $\ell$-qubit random variable that has arbitrary dependence on $\vec{y}$. 
	Recall that we denote $\vec{y}_i = (y_1, \ldots, y_i)$ for $1 \le i \le t$, and $\vec{y}_0$ is the empty vector. 
	Then for a uniformly distributed $J\leftarrow[t]$,  
	\begin{align}
		I(\qv{s}: y_J | \vec{y}_{J-1}, J) \le \frac{\ell}{t}~.
	\end{align}
\end{lemma}
\begin{proof}
	We have that
	\begin{align}
		I(\qv{s}: y_J | \vec{y}_{J-1}, J) & = \Ex_{j} [I(\qv{s}: y_j | \vec{y}_{j-1}) ]\\
		& = \frac{1}{t} \sum_{j \in [t]} I(\qv{s}: y_j | \vec{y}_{j-1}) \\
		\text{(By definition)} \qquad	& = \frac{1}{t} \sum_{j \in [t]} (H(\qv{s} | \vec{y}_{j-1}) - H(\qv{s} | \vec{y}_{j}) )\\
		\text{(Telescopic sum)} \qquad	& = \frac{1}{t} (H(\qv{s}) - H(\qv{s} | \vec{y}))\\
		\text{(Corollary~\ref{cor:posent})} \qquad	& \le \frac{H(\qv{s})}{t}\\
		\text{($\qv{s}$ is $\ell$-qubits)} \qquad	& \le \frac{\ell}{t}~.
	\end{align}
\end{proof}

\begin{proposition}\label{prop:condtd}
	Let $Z$ be a classical variable and let $X,Y$ be quantum variables with arbitrary dependence on $Z$. Then it holds that
	\begin{align}
		\td(XZ, YZ) = \Ex_{z\sim Z}[ \td(X|z, Y|z)]~.
	\end{align}
\end{proposition}
\begin{proof}
	Since $Z$ is classical, the density matrices of $XZ$ and $YZ$ can be written as block-diagonal: $\rho_{XZ}= \sum_z p_z \rho_{X|z} \otimes \ketbra{z}$ and $\rho_{YZ}=\sum_z p_z \rho_{Y|z} \otimes \ketbra{z}$, where $p_z = \Pr[Z=z]$.
	
	Recall that the $\ell_p$ norm of a block-diagonal matrix is simply the sum of norms of the blocks (since each block can be individually diagonalized). We therefore have
	\begin{align}
		\td(XZ, YZ) & = \tfrac{1}{2} \|\rho_{XZ} - \rho_{YZ}\|_1\\
		&= \tfrac{1}{2} \|\sum_z p_z (\rho_{X|z} - \rho_{Y|z}) \otimes \ketbra{z}\|_1\\
		&=\tfrac{1}{2} \sum_z p_z \| \rho_{X|z} - \rho_{Y|z}\|_1\\
		& = \Ex_{z\sim Z}[ \td(X|z, Y|z)]~.
	\end{align}
\end{proof}

We use the following lemma which follows straightforwardly from quantum Pinsker inequality.

\begin{lemma}\label{lem:qpinsker}
	Let $X,Y$ be arbitrary quantum variables with a joint density matrix $\rho_{XY}$ and reduced density matrices $\rho_X, \rho_Y$ respectively.\footnote{Recall that the reduced density matrix of $\rho_{X,Y}$ corresponding to $X$ is $\rho_X=\tr_Y(\rho_{XY})$, where $\tr_Y$ is the linear operator that satisfies that  $\tr_Y(|x_1\rangle\langle x_2|\otimes|y_1\rangle\langle y_2|)= |x_1\rangle\langle x_2|\tr(|y_1\rangle\langle y_2|)$ for any $|x_1\rangle$ and $|x_2\rangle$ in $X$ and $|y_1\rangle$ and $|y_2\rangle$ in $Y$. 
	}
	Then %
	\begin{align}
		\td(\rho_{XY}, \rho_X \otimes \rho_Y) \le \sqrt{ \tfrac{\ln(2)}{2} \cdot I(X:Y)}\le \sqrt{I(X:Y)/2}~,
	\end{align}
	where $\td$ denotes the trace distance.
\end{lemma}
\begin{proof}
	This is a direct application of quantum Pinsker inequality \cite[Theorem 5.38]{WatrousQIT}, 
	when bearing in mind the connection between quantum divergence and mutual information as expressed in \cite[Eq.~(5.110)]{WatrousQIT}.
\end{proof}

We can finally prove the plug-in lemma.
\begin{proof}[Proof of Lemma~\ref{lem:plugin}]
	For convenience, we denote ${y}'_{j} = y'$.
	We start by noticing that by definition, conditioned on $j, \vec{y}_{j-1}$, the value $({y}'_{j}, \qv{s})$ is simply the product distribution of the marginals of $({y}_{j}, \qv{s})$.
	
	Thus,
	\begin{align*}
		\td( (j, \vec{y}_{j-1}, y_j, \qv{s}), (j, \vec{y}_{j-1}, {y}'_{j}, \qv{s}) ) &= \Ex_{j, \vec{y}_{j-1}}\Big[ \td( (\qv{s}, y_j)|(j, \vec{y}_{j-1}), (\qv{s}, {y}'_{j} )|(j, \vec{y}_{j-1}) )\Big]\\
		\text{(Proposition~\ref{prop:condtd})}  & \le \Ex_{j, \vec{y}_{j-1}}\Big[ \sqrt{ \tfrac{1}{2} I ( \qv{s}|(j, \vec{y}_{j-1}) : y_{j}|(j, \vec{y}_{j-1}) )}\Big]\\
		\text{(Lemma~\ref{lem:qpinsker})} & \le \sqrt{ \tfrac{1}{2} \Ex_{j, \vec{y}_{j-1}}[ I ( \qv{s}|(j, \vec{y}_{j-1}) : y_{j}|(j, \vec{y}_{j-1}) )]}\\
		\text{(Convexity (Jensen's Inequality))}& = \sqrt{ \tfrac{1}{2} I ( \qv{s}: y_{j}|j, \vec{y}_{j-1} )]}\\
		\text{(Lemma~\ref{lem:boundmutual})}&\le \sqrt{\ell/(2t)}~.
	\end{align*}
\end{proof}

\newcommand{\etalchar}[1]{$^{#1}$}

\end{document}